\documentclass{pnastwo}
\usepackage[latin9]{inputenc}
\usepackage{color}
\usepackage{float}
\usepackage{amsthm}
\usepackage{amsmath}
\usepackage{amssymb}
\usepackage{mathdots}
\usepackage{graphicx}

\usepackage{algorithm}
\usepackage{algorithmic}
\providecommand{\norm}[1]{\lVert#1\rVert}
\newcommand{\ds}{\displaystyle}
\renewcommand{\qed}{\hfill $\Box$ \medskip}

\def\No{\mathbf{N}^{(1)}}
\def\Nt{\mathbf{N}^{(2)}}

\def\norm#1{{\left\|#1\right\|}}

\def\nm{\noalign{\medskip}}
\def\Tau{\mathcal{J}}
\def\Tauo{\Tau^{(1)}}
\def\Tauoj{\Tau^{(j)}}
\def\Taut{\Tau^{(2)}}
\def\Scl{\mathcal{S}}
\def\Sclo{\Scl^{(1)}}
\def\Scloj{\Scl^{(j)}}
\def\Sclt{\Scl^{(2)}}
\def\Dcrp{\mathcal{I}}
\def\Dcrpo{\Dcrp^{(1)}}
\def\Dcrpt{\Dcrp^{(2)}}

\providecommand{\tabularnewline}{\\}

\theoremstyle{plain}
\newtheorem{thm}{\protect\theoremname}
  \theoremstyle{plain}
  \newtheorem{lem}[thm]{\protect\lemmaname}
  \theoremstyle{plain}
  \newtheorem{prop}[thm]{\protect\propositionname}

\date{}
  \providecommand{\lemmaname}{Lemma}
  \providecommand{\propositionname}{Proposition}
\providecommand{\theoremname}{Theorem}

\begin{document}

\title{Shape recognition and classification in electro-sensing}
\author{Habib Ammari\affil{1}{Department of Mathematics and Applications,
    Ecole Normale Sup\'erieure, 45 Rue d'Ulm, 75005 Paris, France
    (habib.ammari@ens.fr, boulier@dma.ens.fr,
    han.wang@ens.fr)},  Thomas Boulier\affil{1}{},
  Josselin Garnier\affil{2}{Laboratoire de
    Probabilit\'es et Mod\`eles Al\'eatoires \& Laboratoire
    Jacques-Louis Lions, Universit\'e Paris VII, 75205 Paris Cedex 13,
    France (garnier@math.jussieu.fr).},  \and Han Wang\affil{1}{}}

\maketitle

\begin{article}
\begin{abstract}
This paper aims at advancing the field of electro-sensing. It
exhibits the physical mechanism underlying shape perception for
weakly electric fish. These fish orient themselves at night in
complete darkness by employing their active electrolocation
system. They generate a stable, high-frequency, weak electric
field and perceive the transdermal potential modulations caused by
a nearby target with different admittivity than the surrounding
water. In this paper, we explain how weakly electric fish might
identify and classify a target, knowing by advance that the latter
belongs to a certain collection of shapes. The fish is able to
learn how to identify certain targets and discriminate them from
all other targets. Our model of the weakly electric fish relies on
differential imaging, {\it i.e.}, by forming an image from the
perturbations of the field due to targets, and physics-based
classification. The electric fish would first locate the target
using  a specific location search algorithm. Then it could
extract, from the perturbations of the electric field, generalized
(or high-order) polarization tensors of the target. Computing,
from the extracted features, invariants under rigid motions and
scaling yields shape descriptors. The weakly electric fish might
classify a target by comparing its invariants with those of a set
of learned shapes. On the other hand, when measurements are taken
at multiple frequencies, the fish  might exploit the shifts and
use the spectral content of the generalized polarization tensors
to dramatically improve  the stability with respect to measurement
noise of the classification procedure in electro-sensing.
Surprisingly, it turns out that the first-order polarization
tensor at multiple frequencies could be enough for the purpose of
classification. A procedure to eliminate the background field in
the case where the permittivity of the surrounding medium can be
neglected, and hence improve further the stability of the
classification process,  is also discussed.
\end{abstract}


\section{Introduction}

In the turbid rivers of Africa and South America, some species of
fish generate a stable, high frequency, weak electric field
(0.1-10 kHz, $\leq 100$ mV/cm)  which is not enough for defense
purpose. In 1958, Lissmann and Machin \cite{lissmann1958mechanism}
discovered that the emitted electrical signal is in fact used for
active electro-sensing. The weakly electric fish have thousands of
receptor organs at the surface of their skin. A nearby target with
different admittivity than the surrounding water is detected by
measurements of the electric organ discharge modulations at the
receptor organs \cite{bastian, moller1995}. Targets with large
permittivity cause appreciable phase shifts, which can be measured
by receptors called T-type units \cite{book_nelson}. It is an
important input for the fish, and thus it will be the central
point in this paper for shape classification.

Active electro-sensing has driven an increasing number of
experimental, behavioral, biological, and computational studies
since Lissmann and Machin's work
\cite{assad1997electric,babineau2006modeling, budelli2000electric,
chen2005modeling,maciver2001computational,maciver2001prey,rasnow1989simulation,von1993electric}.
Behavioral experiments have shown that weakly electric fish are
able to locate a target \cite{von1993electric} and discriminate
between targets with different shapes \cite{von2007distance}
or/and electric parameters (conductivity and permittivity)
\cite{von1999active}. The growing interest in electro-sensing
could be explained not only by the curiosity of discovering a
``sixth sense'', the electric perception, that is not accessible
by our own senses, but also by potential bio-inspired applications
in underwater robotics. It is challenging to equip robots with
electric perception and provide them, by mimicking weakly electric
fish, with imaging and classification capabilities in dark or
turbid environments \cite{mciver3, boyer, boyer3, boyer2, mciver2,
mciver, von1999active}.

Mathematically speaking, the problem is to locate the target and
identify its shape and material parameters given the current
distribution over the skin.  Due to the fundamental ill-posedness
character of this imaging problem, it is very intriguing to see
how much information weakly electric fish are  able to recover.
The electric field due to the target is a complicated highly
nonlinear function of its shape, admittivity, and distance from
the fish. Thus, understanding analytically this electric sense is
likely to give us insights in this regard \cite{
assad1997electric,babineau2006modeling, boyer4,
budelli2000electric, maciver2001computational,boyer2,
von2007distance}. While locating targets from the electric field
perturbations induced on the skin of the fish is now understood
(see \cite{electroloc, boyer3} and references therein),
identifying and classifying their shapes is considered to be one
of the most challenging problems in electro-sensing. Although the
neuroethology of these fish has significantly been advanced last
years (see \cite{neuro} and references therein), the neural
mechanisms encoding the shape of a target is far beyond the scope
of our study. Rather, this work focuses on the physical
feasibility of such a process, which was not explained until now.

In \cite{electroloc}, a rigorous model for the electro-location of
a target around the fish was derived. Using the fact that the
electric current produced by the electric organ is time-harmonic
with a known fundamental frequency, a space-frequency location
search algorithm was introduced. Its robustness with respect to
measurement noise and its sensitivity with respect to the number
of frequencies, the number of sensors, and the distance to the
target were illustrated. In the case of disk- and ellipse-shaped
targets,  the conductivity, the permittivity, and the size of the
targets can be reconstructed separately from multifrequency
measurements. Such measurements have been used successfully in
trans-admittance scanners of breast tumors \cite{seo1,seo2,
scholz2002towards}.

In this paper, we tackle the challenging problem of shape
recognition and classification. In order to explain how the shape
information is encoded in measured data, we first derive a
multipolar expansion for the perturbations of the electric field
induced by a nearby target in terms of the characteristic size of
the target. Our asymptotic expansion generalizes Rasnow's formula
\cite{rasnowformula} in two directions: (i) it is a higher-order
approximation of the effect of a nearby target and it is valid for
an arbitrary shape and admittivity contrast and (ii) it takes also
into account the fish's body. As it has been first proved in
\cite{electroloc}, by postprocessing the measured data using layer
potentials associated only to the fish's body, one can reduce the
multipolar formula to the one in free space, {\it i.e.}, without
the fish. Then we show how to identify and classify a target,
knowing by advance that the latter belongs to a dictionary of
pre-computed shapes. The shapes considered in this paper have been
experimentally tested and results reported in \cite{gerhard}. This
idea comes naturally in mind when modeling behavioral experiments
such as in \cite{von1999active,von2007distance,von1993electric}.
The pre-computed shapes would then be a model for the fish's
memory (trained to recognize specific shapes), and the experience
of recognition presented here would simulate the discrimination
exercises that are then imposed to them. We develop two algorithms
for shape classification: one based on shape descriptors while the
second is based on spectral induced polarizations. We first
extract, from the data, generalized (or high-order) polarization
tensors of the target (GPTs) \cite{dico}. These tensors, first
introduced in \cite{ammarisima02},  are intrinsic geometric
quantities and constitute the right class of features to represent
the target shapes \cite{AGKLY11, ammarikang2007GPT}. The shape
features are encoded in the high-order polarization tensors. The
extraction of the GPTs can be achieved by a least-squares method.
The noise level in the reconstructed generalized polarization
tensors depends on the angle of view. Larger is the angle of view,
more stable is the reconstruction. $l_1$-regularization techniques
could be used.  Then we compute from the extracted features
invariants under rigid motions and scaling. Comparing these
invariants with those in a dictionary of pre-computed shapes, we
successfully classify the target. Since the measurements are taken
at multiple frequencies, we make use of the spectral content of
the generalized polarization tensor in order to dramatically
improve the stability with respect to measurement noise of the
physics-based classification procedure. In fact, we show
numerically that the first-order polarization tensor at multiple
frequencies is enough for the purpose of classification.

\section{Feature extraction from induced current measurements}

\subsection{Electro-sensing model}
Let us recall the nondimensionalized model of electro-sensing: the
body of the fish is $\Omega$ (of size of order $1$), an open
bounded set in $\mathbb{R}^d, d=2,3$, of class
$\mathcal{C}^{1,\alpha}$, $0<\alpha<1$, with outward normal unit
vector $\nu$. The electric organ is a dipole $f$ placed at
$z_{0}\in\Omega$ or a sum of point sources inside $\Omega$
satisfying the neutrality condition. We refer to \cite{electroloc}
where the equations are nondimensionalized and the different
scales are identified. The fish's skin is very thin and highly
resistive. Its effective thickness, that is, the skin thickness
times the contrast between the water and the skin conductivities,
is denoted by $\xi$ and is of order $10^{-1}$
\cite{assad1997electric}. We assume that the conductivity of the
background medium is $1$ and that its permittivity is vanishing.
Consider a target $D=z+\delta B$, where $\delta \ll 1$ is the
characteristic size of $D$, $z$ is its location, and $B$ a smooth
bounded domain containing the origin. We assume that $D$ is of
complex admittivity $k=\sigma+i\varepsilon\omega$, with $\omega$
being the operating frequency in the range $[1,10]$ and $\sigma$
and $\varepsilon$ being respectively the conductivity and the
 permittivity of the target. It has been also shown in~\cite{electroloc} that, in the
presence of $D$,
 the electric
field $u$ generated by the fish is the solution of the following
system:
\begin{equation}
\left\{ \begin{alignedat}{2}\Delta u & =f & \,\, \mbox{in } \Omega,\\
\nabla\cdot(1+(k-1)\chi_{D})\nabla u & =0 & \,\,   \mbox{in } \mathbb{R}^{d}\setminus\overline{\Omega},\\
\left.\frac{\partial u}{\partial\nu}\right|_{-} & =0 & \,\, \mbox{on } \partial\Omega,\\
{} u |_+ - u|_- & =\xi\left.\frac{\partial u}{\partial\nu}\right|_{+} & \,\, \mbox{on } \partial\Omega\\
\left|u(x)\right| & =O(\vert x\vert^{-d+1}), &
\,\,\left|x\right|\rightarrow\infty.
\end{alignedat}
\right.\label{eq:system-u}
\end{equation}
Here, $\chi_{D}$ is the characteristic function of $D$.
Fig.~\ref{figmodel} shows
 isopotentials with and without a target with zero permittivity but different conductivity from the surrounding
medium. Note that if the target's admittivity depends on the
frequency ({\it i.e.}, if the permittivity is nonzero), then a
phase shift in the electrical potential is induced.

\begin{figure}[!h]
\centering%
\begin{tabular}{cc}
\includegraphics[width=4cm]{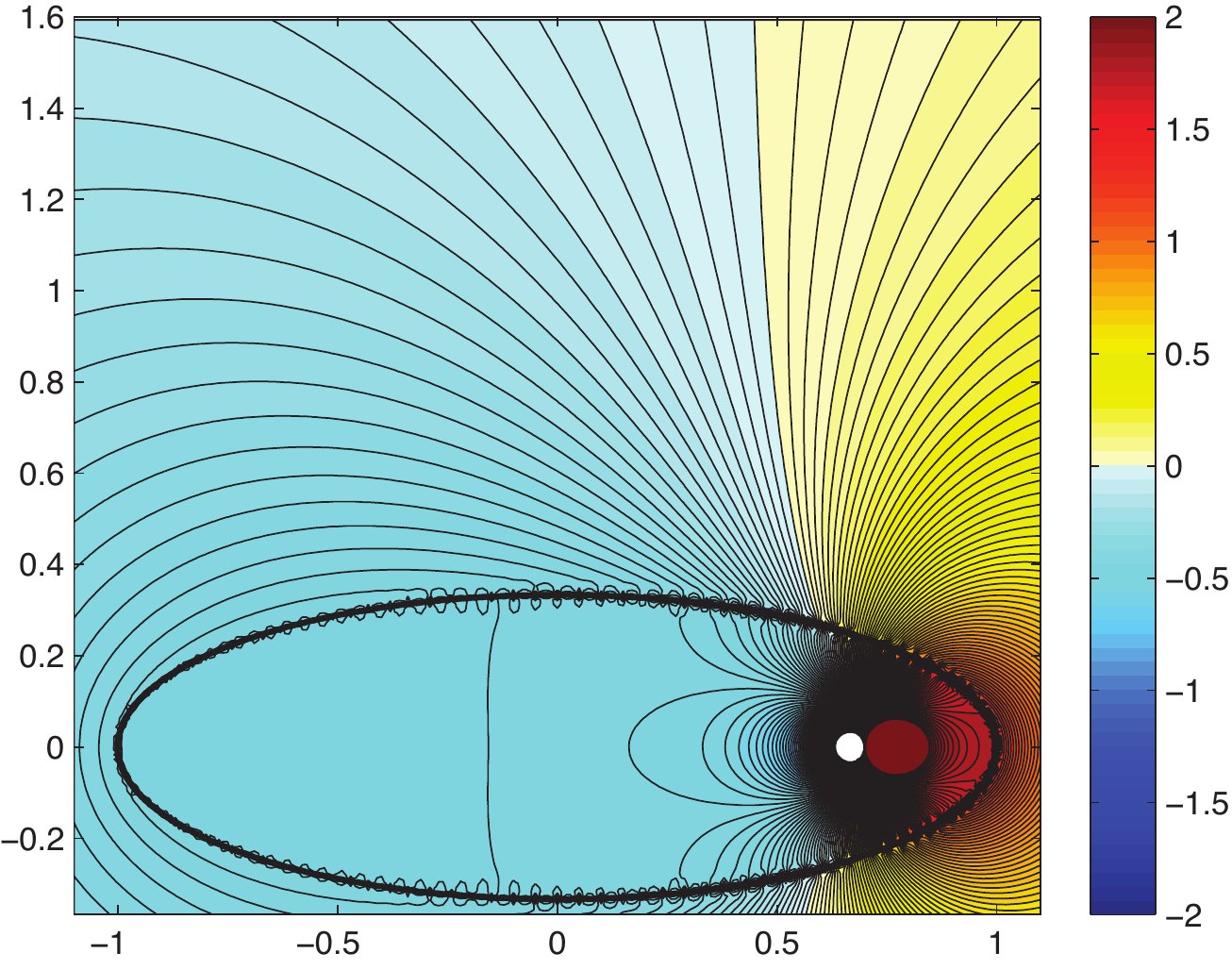} & \includegraphics[width=4cm]{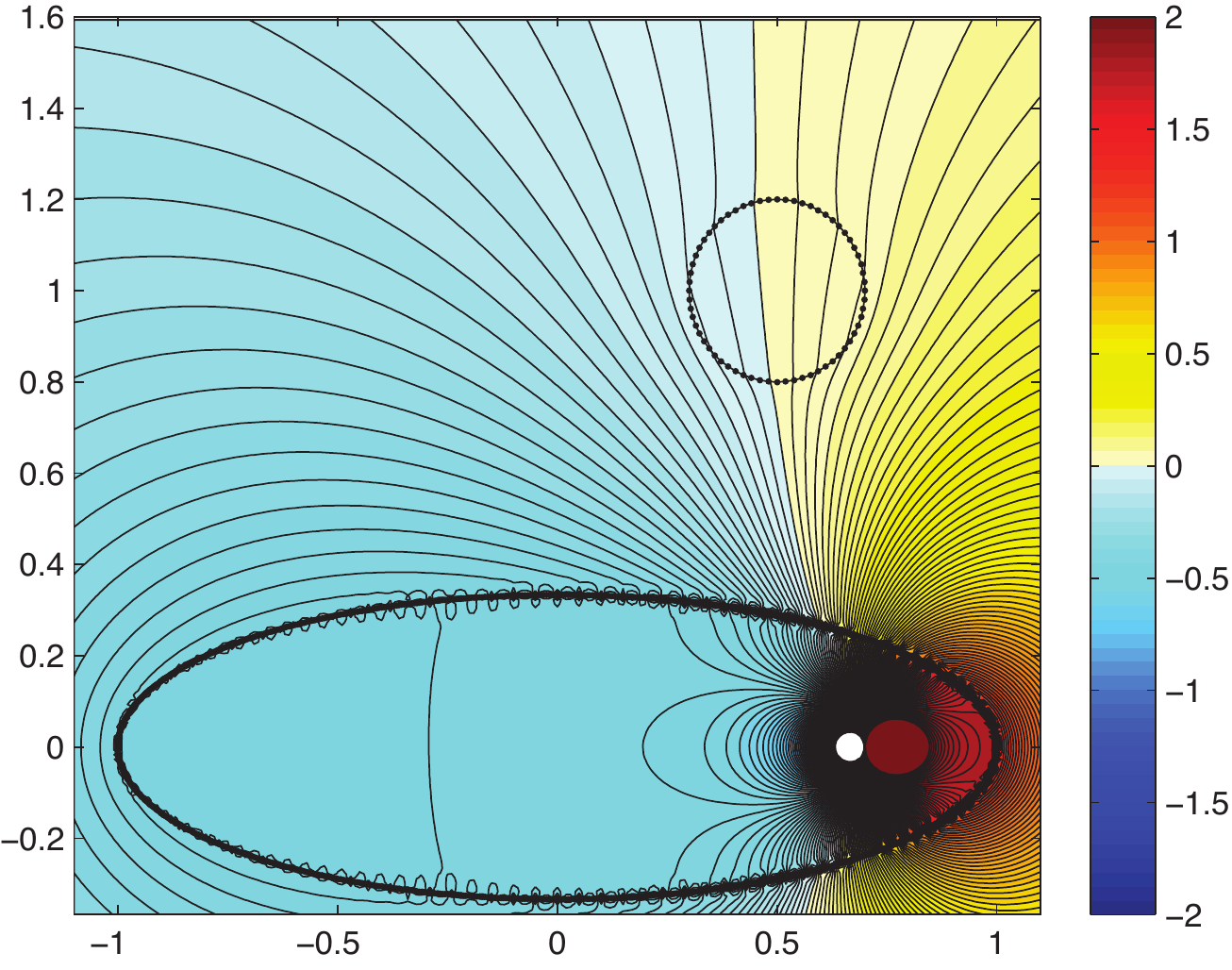}\tabularnewline
(a) & (b)\tabularnewline
\end{tabular}
\caption{\label{figmodel}Isopotentials without (a) and with (b) a
target with $\sigma=5$ and $\varepsilon=0$.}
\end{figure}

In a previous study~\cite{dico}, we have extracted the GPTs of a
target by multistatic measurements using arrays of sources and
receptors. These GPTs were then arranged and compared to a
dictionary of already known shapes. This study aims at adapting
this method to the electro-sensing problem.

%\subsection{Reconstruction of the Generalized Polarization Tensors
%of the target}

\subsection{Asymptotic formalism} \label{sec:reconstruction-GPT}

The first step is to compute the GPTs from the measurements. In
this regard, the next result will be useful. Except when
mentioned, we will fix in this section the frequency~$\omega$,
leading to a fixed complex admittivity~$k$. We will need the
following notation. For a multi-index $\alpha \in \mathbb{N}^d$,
let $x^\alpha = x_1^{\alpha_1} \ldots x_d^{\alpha_d},
\partial^\alpha = \partial_1^{\alpha_1} \ldots
\partial_d^{\alpha_d},$ with $\partial_j = \partial /
\partial x_j$. Let $G(x)$ be the Green function of the Laplacian
in $\mathbb{R}^d$ which satisfies $\Delta G =\delta$ (where
$\delta$ is the Dirac function at the origin) and is given by
$$
G(x) = \left\{
\begin{array}{ll}
\ds \frac{1}{2\pi} \ln |x|, \quad & d=2 , \\
\nm \ds -\frac{1}{4\pi} \frac{1}{|x|}, \quad & d=3.
\end{array}
\right.
$$
We denote the single and double layer potentials of a function
$\phi \in L^2(\partial \Omega)$ as $\mathcal{S}_\Omega[\phi]$ and
$\mathcal{D}_\Omega[\phi]$, where
\begin{align} \label{defs}
\mathcal{S}_\Omega[\phi](x) &:= \int_{\partial \Omega} G(x-y) \phi
(y) \, d \sigma(y), \quad x \in \mathbb{R}^d,
\end{align}
and
\begin{equation} \label{defd}
\mathcal{D}_\Omega[\phi](x) := \int_{\partial \Omega}
\frac{\partial G}{\partial  {\nu(y)}} (x-y) \phi (y) \, d
\sigma(y), \quad x \in \mathbb{R}^d \setminus
\partial \Omega.
\end{equation}
 We also define the  boundary integral operator $\mathcal{K}^*_\Omega$ on
$L^2(\partial \Omega)$ by
\begin{equation} \label{defk}
\mathcal{K}_\Omega^*[\phi](x) := \int_{\partial \Omega} \frac
{\partial G}{\partial {\nu(x)}}(x-y) \phi (y)\,d\sigma(y), \quad
\phi \in L^2(\partial \Omega).
 \end{equation}
The operator $\mathcal{K}^*_\Omega$ is called the
Neumann-Poincar\'e operator. We assume that the target is away
from the fish, {\it i.e.}, the distance between the fish and the
target is much larger than the target's characteristic  size but
smaller than the range of the electrolocation which does not
exceed two fish's body lengths. The following theorem holds.
\begin{thm}
\label{thm:asymptotic-formula}Let us define the function $H:\mathbb{R}^{d}\rightarrow\mathbb{C}$
by
\begin{equation}
H(x)=p(x)+\mathcal{S}_{\Omega}\left[\left.\frac{\partial
u}{\partial\nu}\right|_{+}\right]-\xi\mathcal{D}_{\Omega}\left[\left.\frac{\partial
u}{\partial\nu}\right|_{+}\right],\label{eq:H-def}
\end{equation}
where $p$ is the field created by the dipole $f$, {\it i.e.},
$\Delta p=f$ in $\mathbb{R}^{d}$. Then, for every integer
$K\geq1$, the following expansion holds
\begin{equation} \begin{array}{l}
\ds
u(x)=H(x)+\delta^{d-2}\sum_{\left|\alpha\right|=1}^{K}\sum_{\left|\beta\right|=1}^{K-\left|\alpha\right|+1}
\frac{(-1)^{\vert\beta\vert}\delta^{\left|\alpha\right|+\left|\beta\right|}}{\alpha!\beta!}
\partial^{\alpha}H(z) \\ \nm \qquad \qquad  \ds \times M_{\alpha\beta}(\lambda,B)\partial^{\beta}G(x-z)+O(\delta^{d+K}),
\end{array}\label{eq:asymptotic formula}
\end{equation}
uniformly for $x \in \partial \Omega$, where $z$ is the location
of the target $D$ and
$$M_{\alpha\beta}(\lambda,B):= \int_{\partial B} (\lambda
I-\mathcal{K}_{B}^{*})^{-1}[\frac{\partial y^\alpha}{\partial
\nu}] y^\beta \, d\sigma(y)$$ is the generalized polarization
tensor (of order $(\alpha,\beta)$) associated with the domain $B$
and the contrast $\lambda=(k+1)/2(k-1)$; see
\cite{ammarikang2007GPT}. Here, $\mathcal{K}_{B}^{*}$ is the
Neumann-Poincar\'e operator associated with $B$ and $I$ is the
identity operator.
\end{thm}
Let us make a few remarks. First, the definition of the GPTs still
holds for complex-valued $\lambda$. However, some properties are
lost by this change; thus one has to study them more carefully in
this situation. Second, the function $H$, which is computed from
the boundary measurements,  still depends on $\delta$ but this is
not important for our present study. Indeed, formula
(\ref{eq:asymptotic formula}) could have been derived with $U$,
the background solution in the absence of the target, instead of
$H$ and $G_{R}$ - the Green function associated to Robin
conditions on $\partial\Omega$ - instead of $G$, but it is much
easier to compute $\partial^{\alpha}H(z)$ and
$\partial^{\beta}G(x-z)$ once $z$ is known. This leads us to the
third remark: the location $z$ is supposed to be known from the
algorithm developed in~\cite{electroloc}. Electrolocation
algorithms are either based on a space-frequency approach in the
case of multifrequency measurements or on the fish's movement if
only one frequency is used \cite{electroloc, boyer3}. Finally, it
is worth mentioning that the knowledge of
$M_{\alpha\beta}(\lambda,B)$ for all $\alpha, \beta\in
\mathbb{N}^d$ determines uniquely $B$ and $\lambda$ \cite{mms,
LNM1846}. Moreover, the following scaling relation holds:
$$
M_{\alpha\beta}(\lambda, \delta B) = \delta^{d-2 + |\alpha| +
|\beta|} M_{\alpha\beta}(\lambda,B).
$$

We will follow the proof of \cite[Theorem 4.8]{LNM1846}. In a first
step, let us show the following formula.
\begin{lem}
For $x\in\mathbb{R}^{d}$,
\begin{equation}
u(x)=H(x)+\mathcal{S}_{D}(\lambda
I-\mathcal{K}_{D}^{*})^{-1}\left[\left.\frac{\partial
H}{\partial\nu}\right|_{\partial D}\right](x),\label{eq:u-H}
\end{equation}
where $\nu$ is the outward normal unit vector at $\partial D$ and
$\mathcal{S}_{D}$ and $\mathcal{K}_{D}^{*}$ are respectively
defined by (\ref{defs}) and (\ref{defk}) with $\Omega$ replaced
with $D$.\end{lem}
\begin{proof}
In \cite[Section 4.1.2]{electroloc}, it is shown that
\[
u(x)=p(x)+\mathcal{S}_{\Omega}[\psi](x)+\mathcal{D}_{\Omega}[\phi](x)+\mathcal{S}_{D}[\phi](x),
\]
where the functions $\psi$, $\phi\in L^{2}(\partial\Omega)$ and
$\phi\in L^{2}(\partial D)$ verify the following system
\[
\left\{ \begin{alignedat}{2}\phi & =-\xi\psi & \,\, \mbox{on } \partial\Omega,\\
\left(\frac{I}{2}-\mathcal{K}_{\Omega}^{*}+\xi\frac{\partial\mathcal{D}_{\Omega}}{\partial\nu}\right)[\psi]
-\frac{\partial}{\partial\nu}(\mathcal{S}_{D}[\phi]) & =\frac{\partial p}{\partial\nu} & \,\, \mbox{on } \partial\Omega,\\
-\frac{\partial}{\partial\nu}(\mathcal{S}_{\Omega}[\psi])+\xi\frac{\partial}{\partial\nu}(\mathcal{D}_{\Omega}[\psi])
+(\lambda I-\mathcal{K}_{D}^{*}) & [\phi]=\frac{\partial
p}{\partial\nu} & \,\mbox{on } \partial D.
\end{alignedat}
\right.
\]
The third line gives us
\[
\phi=(\lambda
I-\mathcal{K}_{D}^{*})^{-1}\left[\left.\frac{\partial
H}{\partial\nu}\right|_{\partial D}\right],
\]
and the jump formulas for the single and double layer
potentials~\cite[Theorem 2.17]{ammarikang2007GPT} give us
\[
\psi= \frac{\partial u}{\partial\nu} \bigg|_+ \quad \mbox{ and
}\phi=- u|_+ + u|_-,
\]
so that, from the boundary conditions of the
system~(\ref{eq:system-u}), we obtain
$p+\mathcal{S}_{\Omega}[\psi]+\mathcal{D}_{\Omega}[\phi]=H$ and
the lemma is proved.
\end{proof}

We can now prove Theorem~\ref{thm:asymptotic-formula}, using the
arguments in \cite[pp. 72-73]{LNM1846}. Starting with formula
(\ref{eq:u-H}), the proof relies on a Taylor expansion of $H$ and
the Green function involved in the single layer potential. Indeed,
denoting
\[
H_{K}(x)=\sum_{\left|\alpha\right|=0}^{K}\frac{1}{\alpha!}\partial^{\alpha}H(z)(x-z)^{\alpha},
\]
a Taylor expansion gives us
\[
\left\Vert \frac{\partial H}{\partial\nu}-
\frac{\partial H_{K}}{\partial\nu}\right\Vert _{L^{2}(\partial D)}\leq C\delta^{K}\left|\partial D\right|^{1/2},
\]
and from \cite[Formula (4.10)]{LNM1846}, we have for any $h\in
L^2(\partial D)$ such that $\int_{\partial D}h=0$:
$$
\forall x\in\partial\Omega,\,\,\left|\mathcal{S}_{D}(\lambda
I-\mathcal{K}_{D}^{*})^{-1}[h](x)\right|\leq C\delta
\left|\partial D\right|^{1/2} \left\Vert h\right\Vert
_{L^{2}(\partial D)}.$$ Hence, using the fact that $\left|\partial
D\right| = \delta^{d-1} \left|\partial B\right|$, we obtain
\[\begin{array}{l}
\ds \left\Vert \mathcal{S}_{D}(\lambda
I-\mathcal{K}_{D}^{*})^{-1}\left[\frac{\partial
H}{\partial\nu}-\frac{\partial
H_{K}}{\partial\nu}\right]\right\Vert
_{L^{\infty}(\partial\Omega)} \\
\nm \qquad \leq \ds C\delta\left|\partial D\right|^{1/2}\left\Vert
\frac{\partial H}{\partial\nu}-\frac{\partial
H_{K}}{\partial\nu}\right\Vert _{L^{2}(\partial D)} \leq
C\delta^{d+K}.
\end{array}
\]
Plugging this inequality into (\ref{eq:u-H}) enables us to write,
for $x\in\partial\Omega$,
\[
u(x)=H(x)+\mathcal{S}_{D}(\lambda
I-\mathcal{K}_{D}^{*})^{-1}\left[\frac{\partial
H_{K}}{\partial\nu}\right](x)+O(\delta^{d+K}).
\]
By a change of variables $y'=(y-z)/\delta$,  denoting
$\phi_{\alpha}(y')= (\lambda
I-\mathcal{K}_{B}^{*})^{-1}[\nu\cdot\nabla w^{\alpha}](y')$ for
$y'\in\partial B$ (where $\nu$ is here the outward normal unit
vector to $\partial B$), we have (see for example the arguments
in~\cite[Section 3]{AGKLY11})
\[\begin{array}{l}
\ds
u(x)-H(x)=\sum_{\left|\alpha\right|=0}^{K}\frac{1}{\alpha!}\partial^{\alpha}H(z)
\delta^{\left|\alpha\right|+d-2} \\ \nm \qquad \qquad \ds \times
\int_{\partial B}G(x-z-\delta
y')\phi_{\alpha}(y')d\sigma(y)+O(\delta^{d+K}).
\end{array} \] We can now conclude by injecting a Taylor expansion of
the Green function
\[
G(x-z-\delta y)=\sum_{\left|\beta\right|=0}^{\infty}\frac{(-\delta)^{\left|\beta\right|}}{\beta!}\partial^{\beta}G(x-z)y^{\beta},
\]
in the integrand, giving
\[\begin{array}{l}
\ds
u(x)-H(x)=\delta^{d-2}\sum_{\left|\alpha\right|=0}^{K}\sum_{\left|\beta\right|=0}^{K-\left|
\alpha\right|+1}\frac{(-1)^{\vert\beta\vert}\delta^{\left|\alpha\right|+\left|\beta\right|}}{\alpha!\beta!}
\partial^{\alpha}H(z)\\ \nm \qquad \qquad \ds \times \partial^{\beta}G(x-z)\int_{\partial
B}y^{\beta}\phi_{\alpha}(y)d\sigma(y)+O(\delta^{d+K}).
\end{array} \] The last term is $M_{\alpha\beta}(\lambda,B)$ by
definition~\cite{ammarikang2007GPT}; it then suffices to show that
the terms with $\left|\alpha\right|=0$ or $\left|\beta\right|=0$
vanish, which is the case because $\int_{\partial
B}\phi_{\alpha}=0$ and $\phi_{\alpha}=0$ if
$\left|\alpha\right|=0$. Thus,
Theorem~\ref{thm:asymptotic-formula} is proved. \qed

\subsection{Data acquisition and reduction}

Let us suppose that the fish is moving, and let us take a sample
of $S\in\mathbb{N}^{*}$ different positions
$\left(\Omega_{s}\right)_{1\leq s\leq S}$. This gives us $2S$
different functions $\left(u_{s}\right)_{1\leq s\leq S}$ and
$\left(H_{s}\right)_{1\leq s\leq S}$, leading us to the following
data matrix
\begin{equation}
\mathbb{Q}:=\left(Q_{sr}\right)_{1\leq s\leq S,1\leq r\leq R}:=
\left(u_{s}(x_{r}^{(s)})-H_{s}(x_{r}^{(s)})\right)_{1\leq s\leq S,1\leq r\leq R},\label{eq:def-Q}
\end{equation}
where $(x_{r}^{(s)}\in\partial\Omega_{s})_{1\leq r\leq R}$ are the
receptors of the fish being in the $s^{\mbox{th}}$ position. The
choices of indices are motivated by the fact that the different
positions play the role of sources.

The goal of this subsection is to simplify this data set in order
to extract the GPTs (precisely, the contracted GPTs as it will be
seen below). Indeed, from~(\ref{eq:asymptotic formula}), one has
\begin{equation}\begin{array}{l}
\ds
Q_{sr}=\delta^{d-2}\sum_{\left|\alpha\right|=1}^{K}\sum_{\left|\beta\right|=1}^{K-\left|\alpha\right|+1}(-1)^{\vert\beta
\vert}\frac{\delta^{\left|\alpha\right|+\left|\beta\right|}}{\alpha!\beta!}\partial^{\alpha}H_{s}(z)
\\ \nm \qquad \qquad \ds \times  M_{\alpha\beta}(\lambda,B)\partial^{\beta}G(x^{(s)}_{r}-z)+O(\delta^{d+K}).\label{eq:Qsr}
\end{array} \end{equation} As in~\cite{dico}, we will express this formula
in terms of contracted GPTs (CGPTs), when the dimension of the
space is $d=2$. Let us first recall the definitions of these
contracted GPTs. For a target $B$ with contrast ratio $\lambda$,
knowing the GPTs $M_{\alpha\beta}(\lambda,B)$ for all indices
$\alpha$ and $\beta$ such that $\left|\alpha\right|=m$ and
$\left|\beta\right|=n$ leads us to construct the following
combinations, called CGPTs,
\[
\begin{array}{l} \ds M_{mn}^{cc}
=\sum_{|\alpha|=m} \sum_{|\beta|=n }
a_{\alpha}^{m}a_{\beta}^{n}M_{\alpha\beta},\\
\nm
\ds M_{mn}^{cs}  =\sum_{|\alpha|=m} \sum_{|\beta|=n} a_{\alpha}^{m}b_{\beta}^{n}M_{\alpha\beta},\\
\nm \ds M_{mn}^{sc}
=\sum_{\left|\alpha\right|=m}\sum_{\left|\beta\right|=n}b_{\alpha}^{m}a_{\beta}^{n}M_{\alpha\beta},\\
\nm \ds M_{mn}^{ss}
=\sum_{\left|\alpha\right|=m}\sum_{\left|\beta\right|=n}b_{\alpha}^{m}b_{\beta}^{n}M_{\alpha\beta},
\end{array}
\]
where the real numbers $a_{\alpha}^{m}$ and $b_{\beta}^{m}$ are
defined by the following relation
\[
(x_{1}+ix_{2})^{m}=\sum_{\left|\alpha\right|=m}a_{\alpha}^{m}x^{\alpha}+i\sum_{\left|\beta\right|=m}b_{\beta}^{m}x^{\beta}.
\]
In the polar coordinates, $x=r_{x}e^{i\theta_{x}}$, these
coefficients also verify the following characterization
\[
\sum_{\left|\alpha\right|=m}a_{\alpha}^{m}x^{\alpha}=r_{x}^{m}\cos
m\theta_{x} \mbox{ and
}\sum_{\left|\beta\right|=m}b_{\beta}^{m}x^{\beta}=r_{x}^{m}\sin
m\theta_{x}.
\]
This enables us to show~\cite[Appendix A.2]{dico} that
\begin{equation}
\frac{(-1)^{|\alpha|}}{\alpha!}\partial^{\alpha}G(x)=\frac{-1}{2\pi\left|\alpha\right|}\left(a_{\alpha}^{\left|\alpha\right|}\frac{\cos\left|\alpha\right|\theta_{x}}{r_{x}^{\left|\alpha\right|}}+b_{\alpha}^{\left|\alpha\right|}\frac{\sin\left|\alpha\right|\theta_{x}}{r_{x}^{\left|\alpha\right|}}\right).\label{eq:coeff-G}
\end{equation}
From the definition of $H$, and with the help of the previous
formula, one can prove the following lemma.
\begin{lem}
%Suppose for simplicity that $\xi=0$.
Let the source $f$ be a dipole of moment $\mathbf{p}_{s}$ placed
at $z_{s}$:
\begin{equation} \label{formps}
p_{s}(x)=\mathbf{p}_{s}\cdot\nabla G(x-z_{s}),
\end{equation}
Then, for any $\alpha\in\mathbb{N}^{2}$, there exist two real
numbers $A_{\left|\alpha\right|,s,z}$ and
$B_{\left|\alpha\right|,s,z}$ such that
\[
\frac{1}{\alpha!}\partial^{\alpha}H_{s}(z)=A_{\left|\alpha\right|,s,z}a_{\alpha}^{\left|\alpha\right|}
+B_{\left|\alpha\right|,s,z}b_{\alpha}^{\left|\alpha\right|}.
\]
Moreover, $A_{\left|\alpha\right|,s,z}$ and
$B_{\left|\alpha\right|,s,z}$ can be expressed in the following
way
\[
\begin{array}{lll}A_{m,s,z} & =&\ds \frac{(-1)^{m}}{2\pi}\mathbf{p}_{s}\cdot\left(\begin{array}{c}
\phi_{m+1}(z-z_{s})\\
\psi_{m+1}(z-z_{s})
\end{array}\right)\\ \nm &&\ds -\frac{1}{2\pi m}\int_{\partial\Omega}\left.\frac{\partial u_{s}}{\partial\nu}
\right|_{+}(y)\phi_{m}(y-z)d\sigma(y),\\
\nm &&\ds -\frac{\xi}{2\pi }\int_{\partial\Omega}
\left(\begin{array}{c}
\phi_{m+1}(y-z)\\
\psi_{m+1}(y-z)
\end{array}\right)\cdot \nu_y \left.\frac{\partial u_{s}}{\partial\nu}
\right|_{+}(y)\, d\sigma(y),
\\
\nm B_{m,s,z} &=&\ds
\frac{(-1)^{m}}{2\pi}\mathbf{p}_{s}\cdot\left(\begin{array}{c}
\psi_{m+1}(z-z_{s})\\
-\phi_{m+1}(z-z_{s})
\end{array}\right)\\ \nm && \ds -\frac{1}{2\pi m}
\int_{\partial\Omega}\left.\frac{\partial
u_{s}}{\partial\nu}\right|_{+}(y)\psi_{m}(y-z)d\sigma(y)\\
\nm &&\ds -\frac{\xi}{2\pi }\int_{\partial\Omega}
\left(\begin{array}{c}
\psi_{m+1}(y-z)\\
-\phi_{m+1}(y-z)
\end{array}\right)\cdot \nu_y \left.\frac{\partial u_{s}}{\partial\nu}
\right|_{+}(y)\, d\sigma(y),
\end{array}
\]
where the functions $\phi_{m}$ and $\psi_{m}$ are defined for
$x\in\mathbb{R}^{2}$, $x=(r_{x},\theta_{x})$ in the polar
coordinates, by
\[
\begin{alignedat}{1}\phi_{m}(x) & =\frac{\cos
m\theta_{x}}{r_{x}^{m}},\quad \psi_{m}(x) & =\frac{\sin
m\theta_{x}}{r_{x}^{m}}.
\end{alignedat}
\]
\end{lem}
\begin{proof}
Let us fix $\alpha\in\mathbb{N}^{2}$ and define $m=\vert\alpha\vert$.
Let us recall the definition of $H$, given in (\ref{eq:H-def})
\[
H_{s}(x)=p_{s}(x)+\mathcal{S}_{\Omega_{s}}\left[\left.\frac{\partial
u_{s}}{\partial\nu}\right|_{+}\right]-\xi\mathcal{D}_{\Omega_{s}}\left[\left.\frac{\partial
u_{s}}{\partial\nu}\right|_{+}\right],
\]
where $\Delta p_{s}=f_{s}$ in $\mathbb{R}^{2}$. From
(\ref{formps}) it follows that
\[
\partial^{\alpha}p_{s}(x)=\mathbf{p}_{s}\cdot\nabla\partial^{\alpha}G(x-z_{s}).
\]
Hence, (\ref{eq:coeff-G}) yields
\[\begin{array}{l}
\ds \frac{(-1)^{|\alpha|}}{\alpha!}\partial^{\alpha}p_{s}(z) = \ds
a_{\alpha}^{m}\left[\frac{-1}{2\pi
m}\mathbf{p}_{s}\cdot\nabla\phi_{m}(z-z_{s})\right]
\\ \qquad \qquad \ds +b_{\alpha}^{m}\left[\frac{-1}{2\pi
m}\mathbf{p}_{s}\cdot\nabla\psi_{m}(z-z_{s})\right]. \end{array}
\] Moreover, we have
\[
\begin{alignedat}{1}\nabla\phi_{m} & =-m\left(\begin{array}{c}
\phi_{m+1}\\
\psi_{m+1}
\end{array}\right),\quad
\nabla\psi_{m} & =-m\left(\begin{array}{c}
\psi_{m+1}\\
- \phi_{m+1}
\end{array}\right).
\end{alignedat}
\]
In the same manner, from
\[
\mathcal{S}_{\Omega_{s}}\left[\left.\frac{\partial
u_{s}}{\partial\nu}\right|_{+}\right](x)=\int_{\partial\Omega_{s}}\left.\frac{\partial
u_{s}(y)}{\partial\nu}\right|_{+} G(y-x)d\sigma(y),
\]
we can deduce
\[
\begin{array}{l}
\ds \frac{1}{\alpha!}
\partial^{\alpha}\mathcal{S}_{\Omega_{s}}\left[\left.\frac{\partial
u_{s}}{\partial\nu}\right|_{+}\right](z)  \\ \nm \ds =
\frac{1}{\alpha!} \int_{\partial\Omega_{s}}\left.\frac{\partial
u_{s}}{\partial\nu}\right|_{+}(y)(-1)^{|\alpha|}\partial^{\alpha}G(y-z)d\sigma(y)\\
 \nm \ds =a_{\alpha}^{m}\left(\int_{\partial\Omega_{s}}\frac{-1}{2\pi m}\left.
 \frac{\partial u_{s}}{\partial\nu}\right|_{+}(y)\frac{\cos
 m\theta_{y-z}}{r_{y-z}^{m}}d\sigma(y)\right)\\
\nm \ds
+b_{\alpha}^{m}\left(\int_{\partial\Omega_{s}}\frac{-1}{2\pi
m}\left.\frac{\partial u_{s}}{\partial\nu}\right|_{+}(y)
 \frac{\sin m\theta_{y-z}}{r_{y-z}^{m}}d\sigma(y)\right).
\end{array}
\]
Combining those two equations leads us to the desired result.
\end{proof}
From (\ref{eq:Qsr}), the data matrix is then expressed as follows
\begin{equation}
\begin{alignedat}{1}Q_{rs} & =\sum_{\left|\alpha\right|+\left|\beta\right|=1}^{K+1}
\left(A_{\left|\alpha\right|,s,z}a_{\alpha}^{\left|\alpha\right|}
+B_{\left|\alpha\right|,s,z}b_{\alpha}^{\left|\alpha\right|}\right)\\ &\times M_{\alpha\beta}(\lambda, \delta B)
\frac{-a_{\beta}^{\left|\beta\right|}\cos\left|\beta\right|\theta_{x_{r}^{(s)}-z}
-b_{\beta}^{\left|\beta\right|}\sin\left|\beta\right|\theta_{x_{r}^{(s)}-z}}{2\pi
\left|\beta\right|r_{x_{r}^{(s)}-z}^{\left|\beta\right|}}\\ & +O(\delta^{K+2})\\
 & =\sum_{m+n=1}^{K+1}\underbrace{\left(\begin{array}{cc}
A_{m,s,z} & B_{m,s,z}\end{array}\right)}_{\mathbf{S}_{sm}}\underbrace{\left(\begin{array}{cc}
M_{mn}^{cc} & M_{mn}^{cs}\\
M_{mn}^{sc} & M_{mn}^{ss}
\end{array}\right)}_{\mathbf{M}_{mn}}\\ & \times \underbrace{\left(\begin{array}{c}
\cos n\theta_{x_{r}^{(s)}-z}\\
\sin n\theta_{x_{r}^{(s)}-z}
\end{array}\right)\frac{-1}{2\pi nr_{x_{r}^{(s)}-z}^{n}}}_{\mathbf{G^{(s)}}_{nr}}
+\underbrace{O(\delta^{K+2})}_{{E}_{rs}}.
\end{alignedat}
\label{eq:Qrs-developped}
\end{equation}
Thus, defining the following matrices
\[
\mathbb{M}=\left(\begin{array}{cccc}
\mathbf{M}_{11} & \mathbf{M}_{12} & \ldots & \mathbf{M}_{1K}\\
\mathbf{M}_{21} &  & \iddots & 0\\
\vdots & \iddots & \iddots & \vdots\\
\mathbf{M}_{K1} & 0 & \ldots & 0
\end{array}\right),\,\,\mathbb{E}=\left({E}_{rs}\right)_{1\leq r\leq R,1\leq s\leq S},
\]
the problem is to recover the matrix $\mathbb{M}$ knowing the matrix
\[
\mathbb{Q}=\mathcal{L}(\mathbb{M})+\mathbb{E},
\]
where $\mathcal{L}$ is the linear operator defined by
(\ref{eq:Qrs-developped}).

Therefore,  the CGPTs of the target $D$ can be reconstructed as
the least-squares solution of the above linear system, {\it i.e.},
\begin{equation}
\widehat{\mathbb{M}} = \underset{\mathbb{M} \perp \mathrm{ker }
(\mathcal{L})}{\mbox{argmin}} \| \mathbb{Q} -
\mathcal{L}(\mathbb{M}) \|^2_{F}, \label{eq:lsqr}
\end{equation}
where $\mathrm{ker } (\mathcal{L})$ denotes the kernel of
$\mathcal{L}$ and $\|\cdot\|_F$ denotes the Frobenius norm of
matrices \cite{dico}.

Let us remark that, in the case of multifrequency measurements
$(\omega_{1},\ldots,\omega_{F})$, we can reconstruct
$\left(\widehat{\mathbb{M}}^{(f)}\right)_{1\leq f\leq F}$ from
$\left(\mathbb{Q}^{(f)}\right)_{1\leq f\leq F}$ analogously.

It is worth mentioning that the location of the target detected by
the fish may be different from the true one. Moreover, the target
may be rotated and hence, the reconstructed CGPTs correspond to a
translated, scaled, and rotated target $B$. In order to recognize
the shape $B$, it is therefore fundamental for the recognition
procedure to have size invariance, rotational invariance, and
translational invariance. This could be related to the behavioral
experiments that have shown that weakly electric fish categorize
targets according to their shapes but not according to  sizes,
locations, or orientations \cite{gerhard}.

\section{Recognition and classification}

\label{sub:identification-algo}

Depending on whether we consider multifrequency measurements or
not, we will not identify the CGPTs in the same way.

\subsection{Fixed frequency setting: shape descriptor based classification}

In~\cite{dico}, an algorithm based on shape descriptors was
developed for the recognition of a target in a more classical
electrical sensing setup, ({\it i.e.}, multiple sources/receptors
placed on the surface of a disk). In this paper, we apply this
algorithm  in the context of electro-sensing.

We recall here the concept and properties of shape descriptors in
two dimensions \cite{dico}. For a double index $mn$, with
$m,n=1,2,\ldots$, we introduce the following complex combinations,
$\No = (\mathcal{N}^{(1)}_{mn})_{m,n}, \Nt =
(\mathcal{N}^{(2)}_{mn})_{m,n}$,  of CGPTs:
\begin{equation}
\begin{aligned}
\mathcal{N}^{(1)}_{mn}(\lambda, D) &= (M^{cc}_{mn} - M^{ss}_{mn})
+
i(M^{cs}_{mn} + M^{sc}_{mn}),\\
\mathcal{N}^{(2)}_{mn}(\lambda, D) &= (M^{cc}_{mn} + M^{ss}_{mn})
+ i(M^{cs}_{mn} - M^{sc}_{mn}).
\end{aligned}
\label{eq:Mccdef}
\end{equation}

Let
$$u=\frac{\mathcal{N}^{(2)}_{12}(D)}{2\mathcal{N}^{(2)}_{11}(D)},
\quad T_{-u}D = \{x-u, x\in D\}.$$ We define  the following
quantities which are translation invariant:
\begin{align}
  \label{eq:shape_descrp_trans}
  {\mathbf \Tauo}(D) &= \No(T_{-u}D) = \mathbf{C}^{-u}\No(D)(\mathbf{C}^{-u})^t, \\
  \nm
  {\mathbf \Taut}(D) &= \Nt(T_{-u}D) = \overline{\mathbf{C}^{-u}}\Nt(D)(\mathbf{C}^{-u})^t,
\end{align}
with $t$ being the transpose and the matrix $\mathbf{C}^{-u}$
being a lower triangular matrix with the $m,n$-th entry given by
  $$
  \mathbf{C}^{-u}_{mn}= \binom{m}{n}  (-u)^{m-n}.
  $$
From ${\mathbf \Tauo}(D) =(\Tauo_{mm}(D))_{m,n}$ and ${\mathbf
\Taut}(D) =(\Taut_{mm}(D))_{m,n}$, we define, for any indices
$m,n$, the scaling invariant quantities:
\begin{align}
  \label{eq:shape_descrp_scl}
  \Scloj_{mn}(D) =
  \frac{\Tauoj_{mn}(D)}{\left(\Taut_{mm}(D)\Taut_{nn}(D)\right)^{1/2}},
  \quad j=1,2.
 % \Sclt_{mn}(D) =
 % \frac{\Taut_{mn}(D)}{\left(\Taut_{mm}(D)\Taut_{nn}(D)\right)^{1/2}}
  %.
\end{align}
Finally, we introduce the CGPT-based shape descriptors ${\mathbf
\Dcrpo} = ( \Dcrpo_{mn})_{m,n}$ and ${\mathbf \Dcrpt} =
(\Dcrpt_{mn})_{m,n}$:
\begin{align}
  \label{eq:shape_descrp}
  \Dcrpo_{mn}(D) = |\Sclo_{mn}(D)|, \qquad \ \Dcrpt_{mn}(D) =
  |\Sclt_{mn}(D)|,
\end{align}
where $|\cdot|$ denotes the modulus of a complex number.
Constructed in this way, ${\mathbf \Dcrpo}$ and ${\mathbf \Dcrpt}$
are invariant under translation, rotation, and scaling. Only shape
descriptors of order $2$, {\it i.e.}, for $m,n=1,2$ will be used
in the sequel. Shape descriptors in three dimensions were derived
in \cite{invariants3D}.

\subsection{Multifrequency setting: Spectral induced polarization based classification}

When multiple frequencies are involved, we can use the shape
descriptors $\Dcrpo_{mn}(D)$ and $\Dcrpt_{mn}(D)$ at frequencies
$\omega_1, \ldots, \omega_F$ to enhance the stability of the
classification. However, as it will be shown later, this does not
yield a very stable classification procedure.

Here we rather focus on the first-order polarization tensor (PT),
that is,  the  $2\times2$ complex matrix $\mathcal{M}^{(f)}(D)$
associated with the target $D$ and frequency $f$:
$$
\mathcal{M}^{(f)}(D):= \int_{\partial D} \bigg(\frac{\sigma + 1 +
i \omega_f \varepsilon}{2(\sigma -1 + i \omega_f \varepsilon)}
I-\mathcal{K}_{D}^{*} \bigg)^{-1}[\nu] y \, d\sigma(y),
$$
for $f=1,\ldots, F$. We will show that they are sufficient to
identify  efficiently the targets. Note that it is not possible to
compute the shape descriptors $\Dcrpo_{mn}(D)$ and
$\Dcrpt_{mn}(D)$ based only on first-order PT, because they
require at least second-order polarization tensors. This limits
the use of shape descriptors in the limited-view case where the
reconstruction of higher-order GPTs is not accurate
\cite{ammari2012tracking}.

Here  we use the spectral content of the first-order PTs for
recognition. We have the following
properties~\cite{ammarikang2007GPT}.
\begin{prop}
For any scaling $\delta>0$, rotation angle $\theta\in\mathbb{R}$
and translation vector $z\in\mathbb{R}^{2}$, let us denote
\[
D= z+\delta R_{\theta} B:=\left\{ x=z+\delta R_{\theta}u,\,\, u\in
B\right\} ,
\]
where
\[
R_{\theta}:=\left(\begin{array}{cc}
\cos\theta & -\sin\theta\\
\sin\theta & \cos\theta
\end{array}\right),
\]
is the rotation matrix of angle $\theta$. Then,
\begin{equation} \label{eqrota}
\mathcal{M}^{(f)}(D)=\delta^{2}R_{\theta}\mathcal{M}^{(f)}(B)R_{\theta}^{T}.
\end{equation}
Hence, if we denote by $\tau_{1}^{(f)}(D)$ and $\tau_{2}^{(f)}(D)$
the singular values of $\mathcal{M}^{(f)}(D)$, we obtain
\[
\forall
j\in\{1,2\},\,\,\tau_{j}^{(f)}(D)=\delta^{2}\tau_{j}^{(f)}(B).
\]

\end{prop}
This gives an idea for two algorithms:
\begin{enumerate}
\item The first one, matching the singular values of all the
first-order PT $\left(\mathcal{M}^{(f)}\right)_{1\leq f\leq F}$,
would be dependent of the characteristic scale $\delta$ of the
targets in the dictionary; \item The second one, independent of
the scale of the target, would match the following quantities
\begin{equation}
\mu_{j}^{(f)}=\frac{\tau_{j}^{(f)}}{\tau_{j}^{(F)}},\label{eq:ratio_sv}
\end{equation}
for $j=1,2$ and $f=1,\ldots,F-1$.
\end{enumerate}
Some comments are in order. First, the reason why we consider the
first one, even if it is scale-dependent, is because it is far
more stable. Also, in some biological experiments, two targets of
different scales are considered as
different~\cite{von2007distance}. A question raised was then: how
is it possible to discriminate between a nearby small target and
an extended one situated far away? With the second algorithm, we
have an answer. The last remark concerns
equation~(\ref{eq:ratio_sv}). We could have also considered other
scale-dependent ratios, such as
\[
\frac{\tau_{j}^{(f)}}{\tau_{j}^{(1)}}\mbox{ or
}\frac{\tau_{j}^{(f)}}{\sum_{f'=1}^{F}\tau_{j}^{(f')}},
\]
but since $\tau_{j}^{(F)}$ happens to be the largest one (the
frequencies are sorted in increasing order), it is more stable to
consider~(\ref{eq:ratio_sv}). It is worth mentioning that if there
exists an integer $p>2$ such that $R_{2\pi/p} D = D$, then
$\mathcal{M}^{(f)}(D)$ is proportional to identity.

\subsection{Background field elimination}

We can also improve stability of reconstruction by eliminating the
background field. Let us denote by $U(x)$ the background electric
field ({\it i.e.}, the solution of (\ref{eq:system-u}) with
$k=1$). In~\cite[Proposition 2]{electroloc}, we have proved the
following formula:
\begin{equation}
\mathcal{P}_\Omega \left(  \left.\frac{\partial u_f}{\partial
\nu}\right|_+ - \left.\frac{\partial U}{\partial \nu}\right|_+
\right) \approx \nabla U(z)^T \mathcal{M}^{(f)}(D) \nabla_z \left(
\frac{\partial G}{\partial \nu_x}\right),
\label{eq:approx-dipol-U}
\end{equation}
where $u_f$ is $u$ associated with the $f^{\textrm{th}}$
frequency, $\mathcal{M}^{(f)}(D)$ is the first-order polarization
tensor at the  $f^{\textrm{th}}$ frequency, and
$\mathcal{P}_\Omega$ is the (real-valued) postprocessing operator
given by
$$
\mathcal{P}_\Omega:= \frac{1}{2} I - \mathcal{K}_\Omega^* - \xi
\frac{\partial \mathcal{D}_\Omega}{\partial \nu},
$$
with $\mathcal{D}_\Omega$ and $\mathcal{K}_\Omega^*$  being
defined by (\ref{defd}) and (\ref{defk}); see \cite{electroloc}.
Hence, if the emitted signal $U$ is real-valued, then taking the
imaginary part leads us to
\begin{equation}
\mathcal{P}_\Omega \left[ \Im m \, \left( \left.\frac{\partial
u_f}{\partial \nu}\right|_+  \right) \right] \approx \nabla U(z)^T
\Im m \,  \mathcal{M}^{(f)}(D)  \nabla_z \left( \frac{\partial
G}{\partial \nu_x}\right). \label{eq:approx-dipol-imag}
\end{equation}
Note that in biological sciences, the restriction on $U$ to be
real is justified since the permittivities of water and the fish
are negligible \cite{electroloc}. Now, from
(\ref{eq:approx-dipol-imag}), we can extract $\Im m\,
\mathcal{M}^{(f)}(D)$ by solving a least-squares problem similar
to (\ref{eq:lsqr}). Then, we have the singular values of the
imaginary part of $\mathcal{M}^{(f)}(D)$, which would be
sufficient for shape recognition and classification. The goal of
this procedure is to get rid off the computation of $
\partial U/\partial \nu$ in (\ref{eq:approx-dipol-U}), which is
supposed to be performed numerically in real-world applications,
thus subject to errors. Note that the postprocessing operator
${\mathcal P}_\Omega$ makes the data independent of the shape of
the fish's body.

Because of the following relation which follows from
(\ref{eqrota})
\begin{equation*}
\mathcal{M}^{(f)}(D) = O(\delta^2 \mathcal{M}^{(f)}(B)),
\end{equation*}
taking the imaginary part would lead us to only compute  $\nabla
U(z)$ in (\ref{eq:approx-dipol-imag}) and hence, the error made
here would be modulated by a factor of order $\delta^2$.

\section{Numerical illustrations}

In this section, we illustrate the performance of the algorithms
developed in the previous section. We use the CGPTs obtained in
order to classify the targets. We present an example with fixed
frequency, and another with multifrequency measurements. As it
will be seen, the latter does not lead us to a significantly more
stable classification in the presence of noise or for limited
aperture. The errors in the reconstruction of the high-order
polarization tensors due to measurement noise or the limited-view
aspect deteriorate the stability of the proposed algorithm.
However, when, at multiple frequencies, only the first-order
polarization tensor is used, we arrive at a very robust and
efficient classification procedure.

For the sake of clarity, and due to the large numbers of computations
performed, the results are presented in the appendix.

\subsection{Setup and methods}
We describe the dictionary as well as the measurement systems. We
consider two different shapes for the fish: ellipses and twisted
ellipses. Note that this variety of shapes exists in nature. On
the one hand, twisted ellipses would represent electric eels
(\emph{Electrophorus electricus}), whereas on the other hand
straight ellipses would look like \emph{Apteronotids}
\cite{moller1995}. This simplified representation shows that the
principle of our algorithms can be generalized to any kind of
fish's shape (hence modeling, for example, electro-sensing for
\emph{Mormyrids} as well). It also enhances the fact that, for
bio-inspired engineering applications, the shape of the robot is
not determining. Moreover, as we will see later, our simplified
representation is a good model to tackle aperture issues.

\subsubsection{Dictionary}

The dictionary $\mathcal{D}$ that we consider is composed by 8
different targets: a disk, an ellipse, the letter 'A', the letter
'E', a rectangle, a square, a triangle, and an ellipse with
different electrical parameters (see Fig.~\ref{fig:dico}). Indeed,
all the other targets have conductivity $\sigma=2$ and
permittivity $\varepsilon=1$ whereas the last one has conductivity
$\sigma=5$ and permittivity $\varepsilon=2$. Except when
mentioned, the characteristic size of the target will not matter,
and will be fixed to be of order~$1$.

\begin{figure}[!h]
\centering\includegraphics[width=9cm]{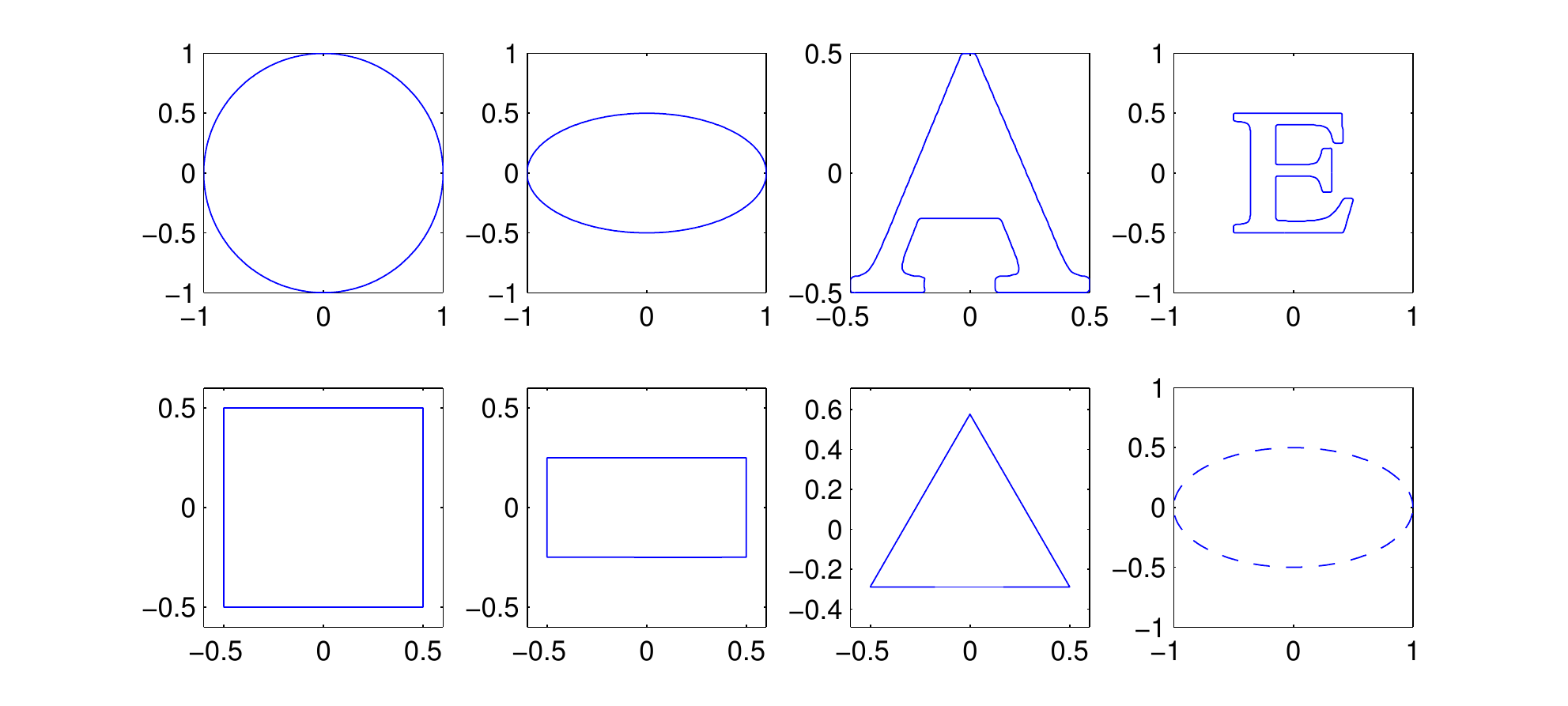}

\caption{\label{fig:dico}The 8 elements of the dictionary. The
dotted lines indicate a target with different electrical
parameters.}

\end{figure}

\subsubsection{Measurements}

In each numerical experiment, one target of the dictionary is
placed at the origin, while the fish swims around it. As it has
been mentioned, we consider two different shapes for the fish's
body: ellipses and twisted ellipses. The measured data is built
taking $20$ positions of the fish around the target (see
Fig.~\ref{fig:Setup}).
\begin{figure}[!h]
\centering%
\begin{tabular}{cccc}
\includegraphics[width=2cm]{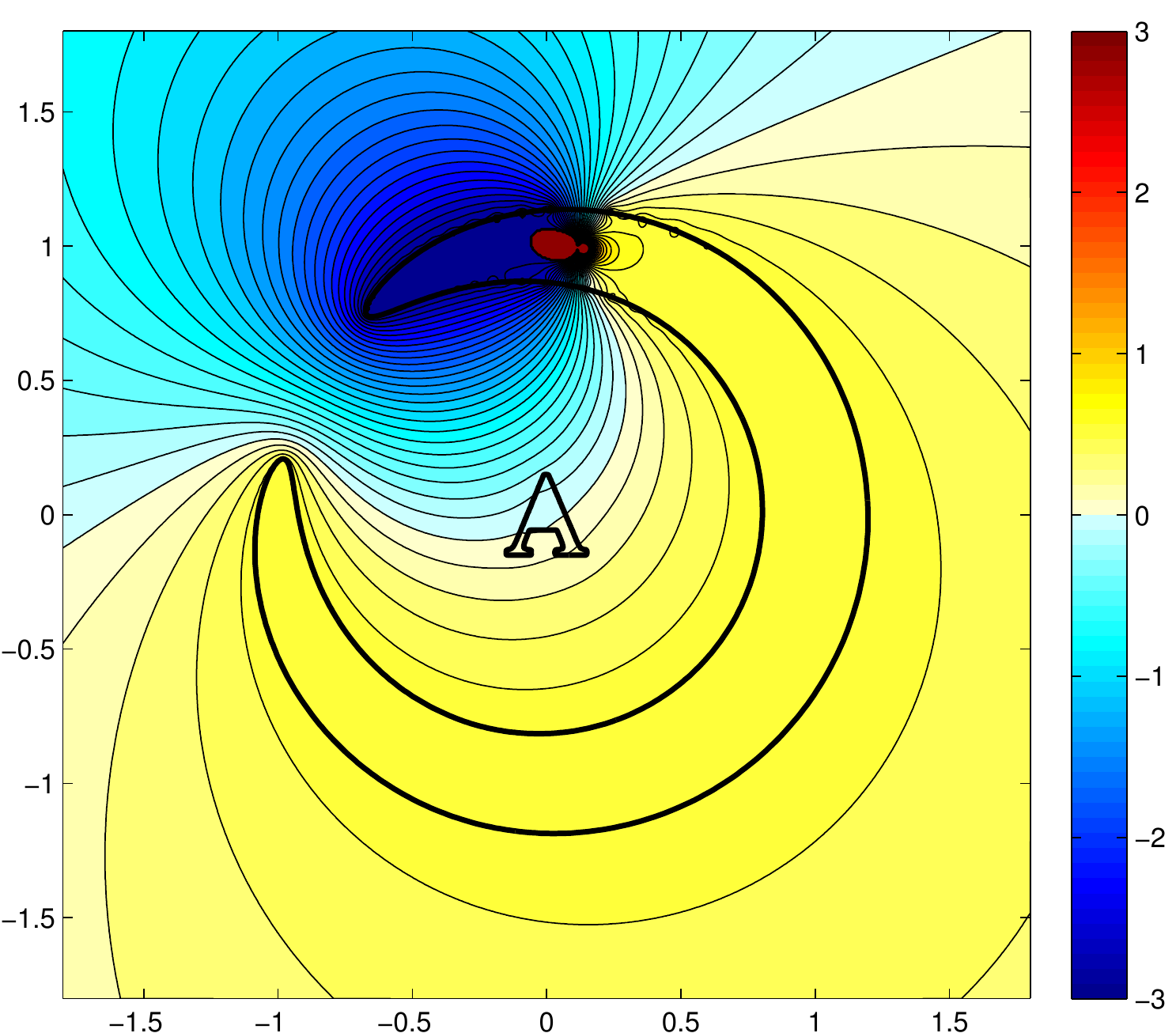} & \includegraphics[width=2cm]{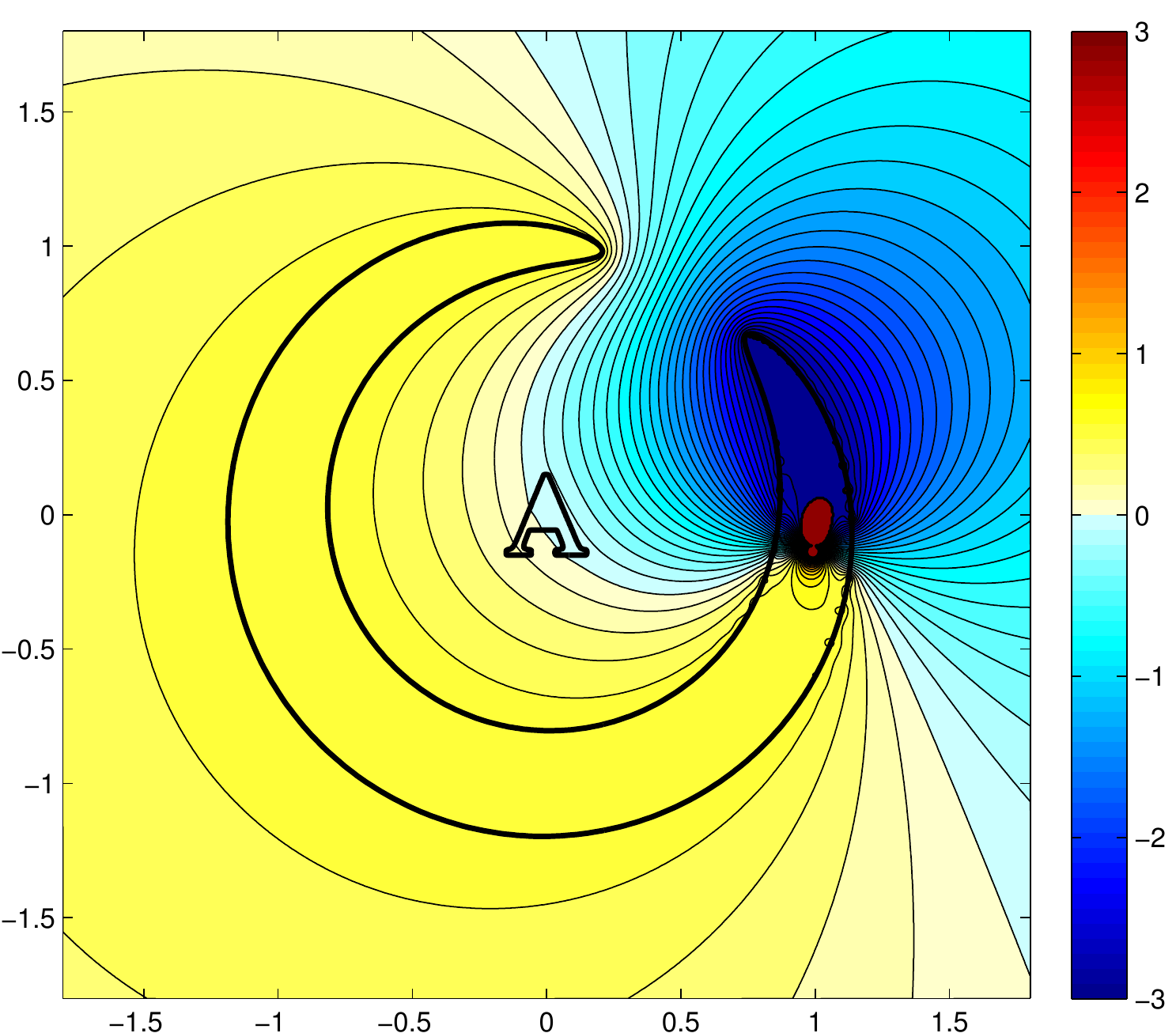} &\includegraphics[width=2cm]{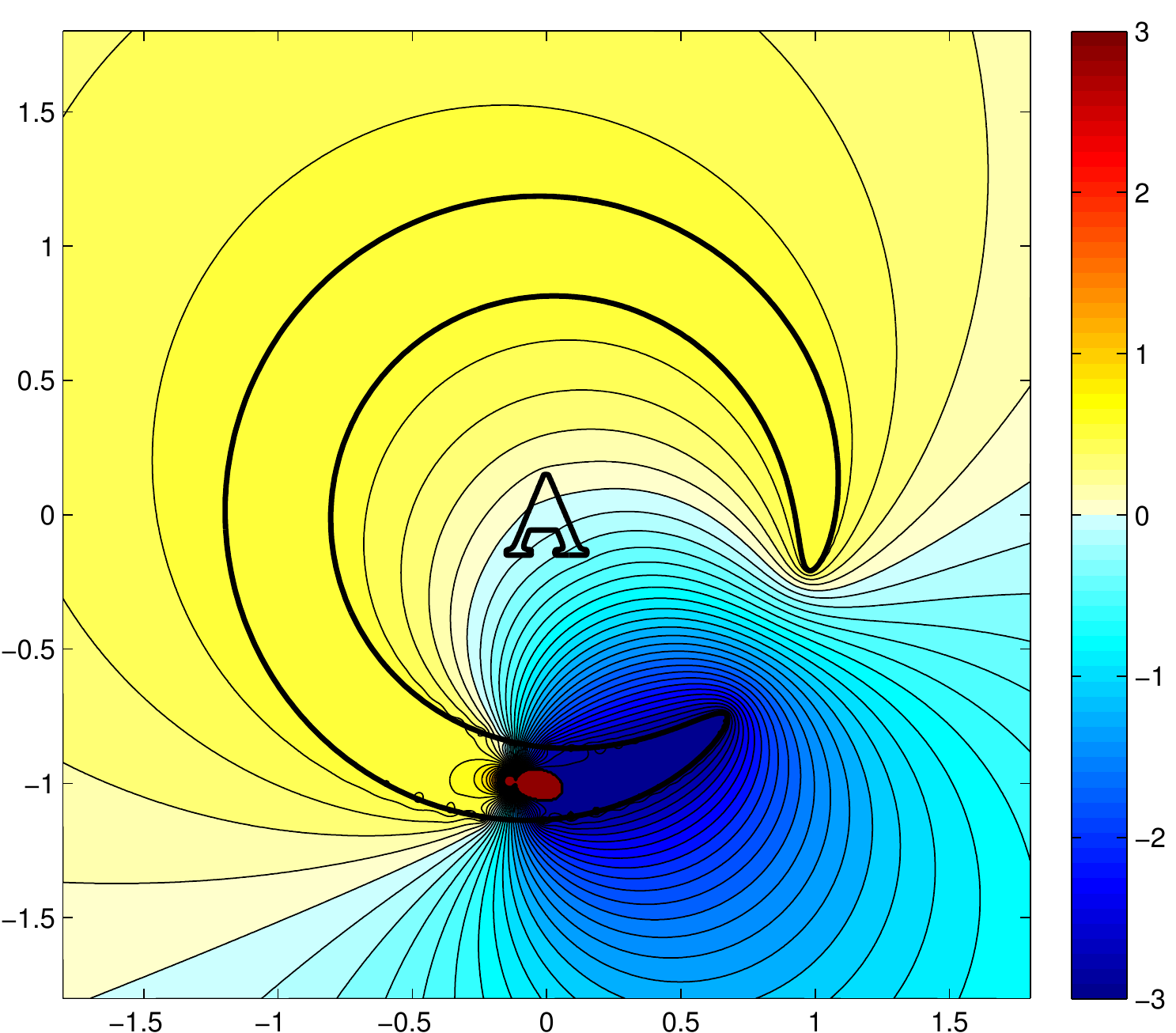}
\\ \tabularnewline
\includegraphics[width=2cm]{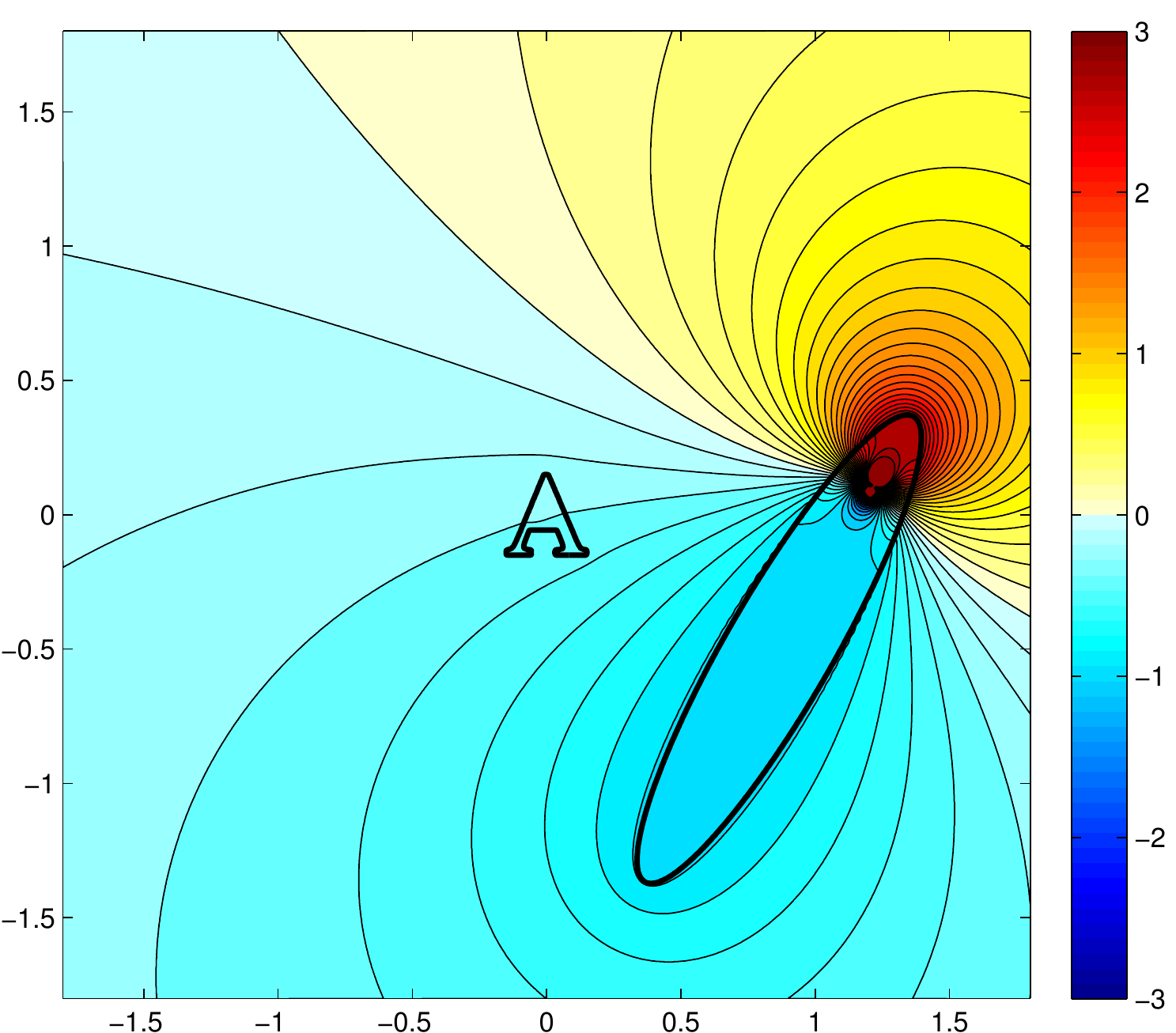} &
\includegraphics[width=2cm]{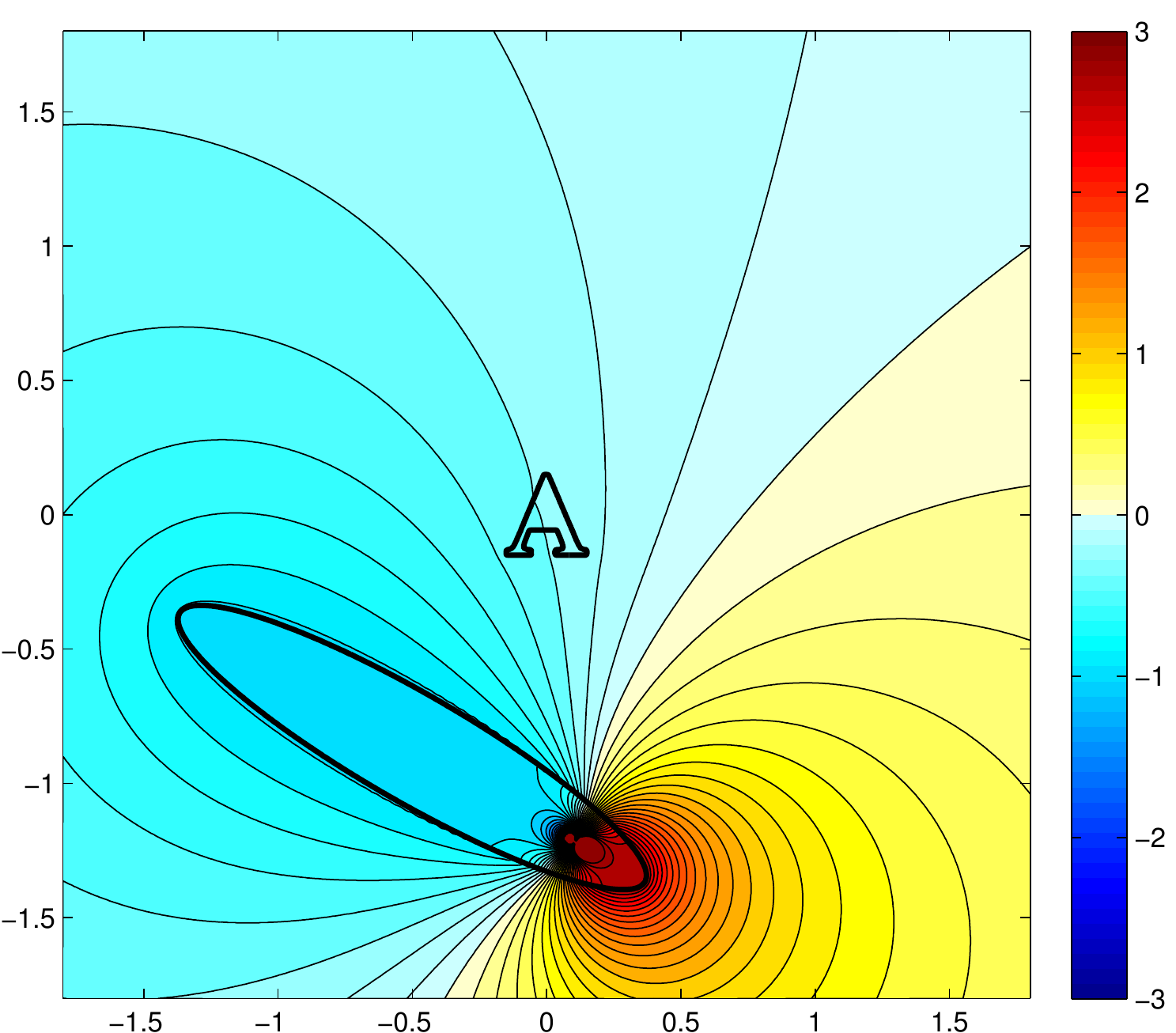} & \includegraphics[width=2cm]{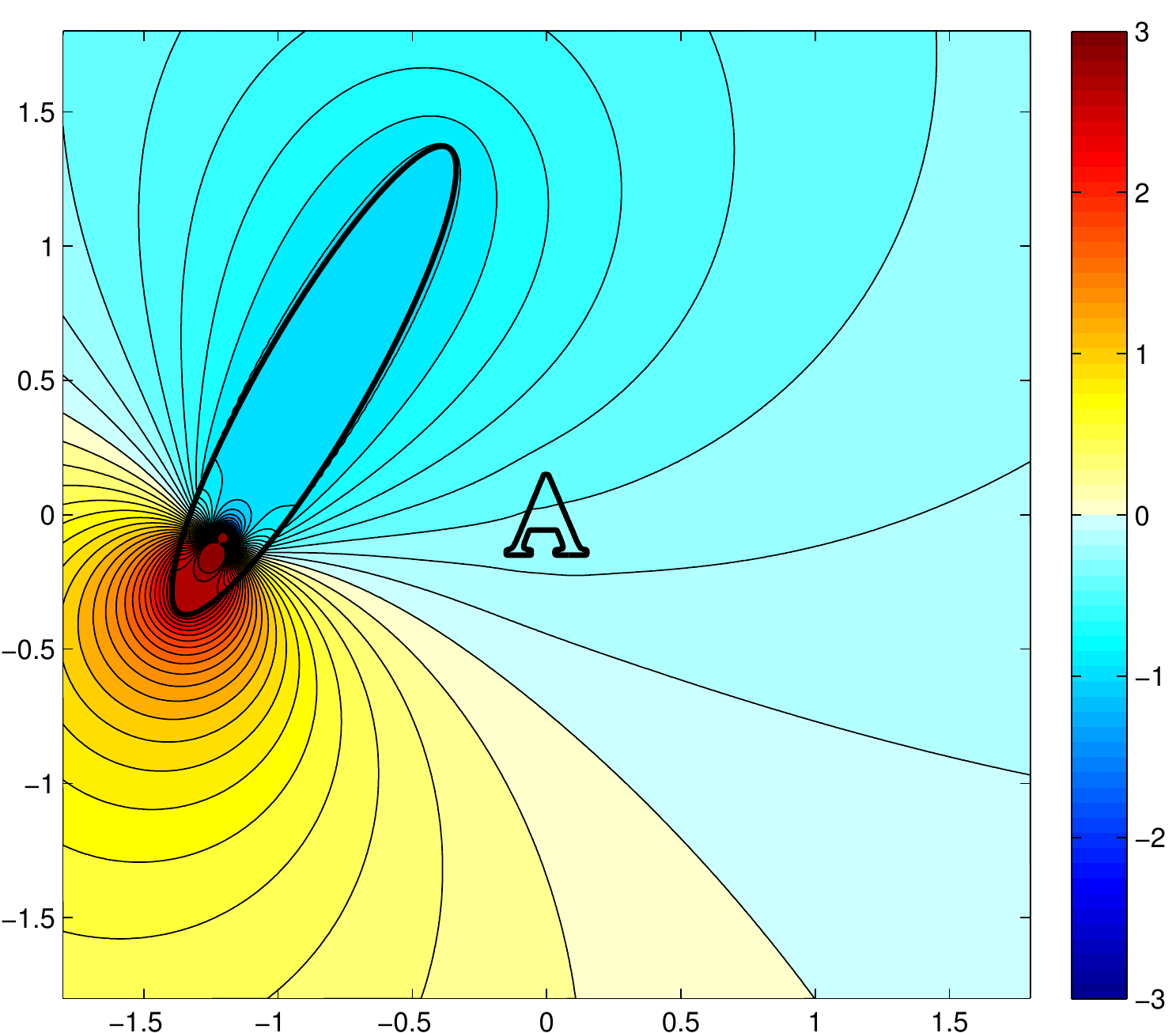}
\tabularnewline
\end{tabular}
\caption{\label{fig:Setup}Two different kinds of experiences
involve (on the top) a twisted-ellipse shape or (on the bottom)
ellipse shape. The real part of the electric field is plotted, for
$3$ (out of $20$) positions that the fish takes around the target
(placed at the origin).}
\end{figure}

In Fig. \ref{fig:aperture}, a smaller aperture than the one in
Fig. \ref{fig:Setup} is considered.

\begin{figure}
\centering
\includegraphics[width=5.2cm, height=4.1cm]{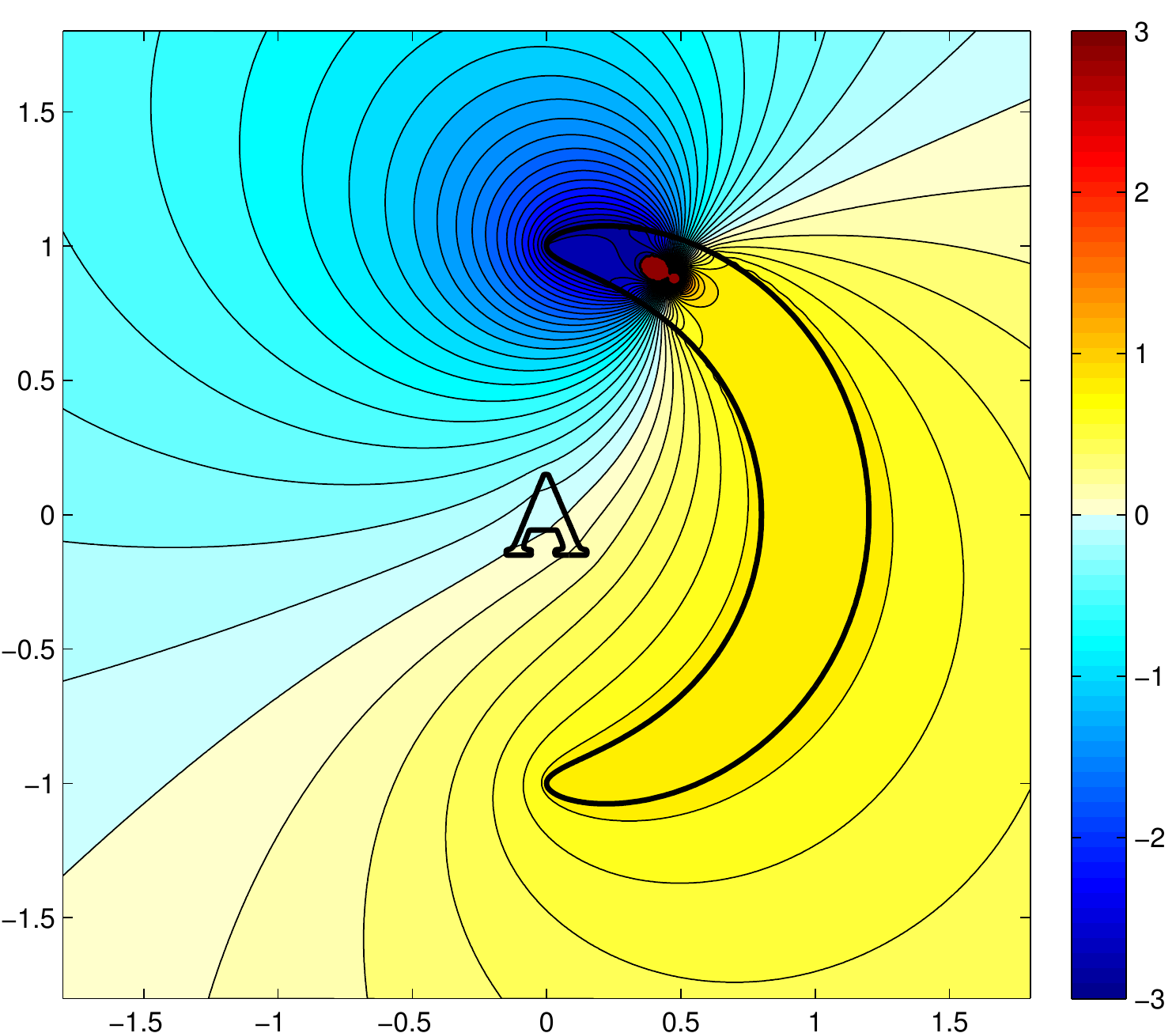}
\caption{\label{fig:aperture} A fish having the shape of a twisted
ellipse with a smaller aperture than in Fig. \ref{fig:Setup}.}
\end{figure}

The typical size of the target is $\delta=0.3$ while the fish is
turning around a disk of radius $R=1$; the twisted ellipse's
semi-axes are $a=1.8$ and $b=0.2$ while the straight ellipse's
semi axis are $a=1$ and $b=0.2$. The effective thickness  of the
skin is set at $\xi=0$. The fish has $2^{7}$ receptors uniformly
distributed on its skin, and the electric organ emits $F=10$
frequencies, equally distributed from $\omega_{0}:=1$ to
$F\omega_{0}$. Again, we refer to \cite{electroloc} for
nondimensionalization of the underlying model equations with
proper quantities and realistic values.

\subsubsection{Classification}

The recognition process is as follows. When measurements are
acquired, we perform least-square reconstruction of the (first- or
second-order) CGPT of the targets. From this CGPT, we compute
quantities of interest $q$ (\emph{i.e.} Shape Descriptors or
singular values or ratios of singular values). Then, for each
element $n$ in the dictionary $\mathcal{D}$, we compute
$\norm{q-q_n}$, where $q_n$ is the -~pre-computed~- quantity of
interest for the $n$th shape. This leads us to charts such as the
ones presented in Figs.~\ref{fig:results-SD-banana},
\ref{fig:results-SD-ellipse}, and \ref{fig:results-aperture}.

  \begin{itemize}
    \item[]{Framework for algorithms of multifrequency classification:}
    \begin{algorithmic}[1]
      \STATE Input: the quantities of interest $\left(q^{(f)}\right)_{1\leq f\leq F}$ calculated from
      the measurement of an
      unknown shape $D$;
      \FOR {$B_n\in \mathcal{D}$}
      \STATE $e_n\leftarrow \sum_{1\leq f \leq F} \| q(B_n)^{(f)} - q^{(f)} \|^2 $ where $\left( q(B_n)^{(f)}
      \right)_{1\leq f \leq F}$ is the same type of quantities of interest of the shape $B_n$;
      \STATE $n\leftarrow n+1$;
      \ENDFOR
      \STATE Output: the true dictionary element $n^*\leftarrow \mbox{argmin}_n e_n$.
    \end{algorithmic}
  \end{itemize}

\subsubsection{Stability analysis}

First, let us explain what kind of noise is considered. We will
add a random matrix (with Gaussian entries) to the data matrix $
$$\mathbb{Q}$ defined in~(\ref{eq:def-Q}). More precisely, we will
consider
\[
\tilde{\mathbb{Q}}:=\mathbb{Q}+\varepsilon\mbox{\ensuremath{\mathbb{W}}},
\]
where $\mathbb{W}$ is a $S\times R$ matrix whose coefficients
follow a Gaussian distribution with mean $0$ and variance $1$. The
real number $\varepsilon$ is the strength of the noise, and will
be given in percentage of the fluctuations of $\mathbb{Q}$, ({\it
i.e.},  $\max_{s,r}Q_{sr}-\min_{s,r}Q_{sr}$). The recognition
procedure remains the same.

Stability analysis was then carried out empirically: for each
noise level, we performed $N_{\textrm{stabil}}$ independent
recognition
 process (with $N_{\textrm{stabil}}$ being precised for each experiment), and computed the ratio of good detection. The computation ends when we reach the threshold of $12.5\%$ probability of detection that corresponds to a
 uniform random choice of the object. This gives us Figs.~\ref{fig:stability-SD-mono} to
 \ref{fig:results-imag-ellipse-ratio}.

\subsection{Results and discussion}

\label{sub:results}

In this subsection, we compare the respective stability of our algorithms. Due to the large number of situations and computations, only significant figures were plotted, giving the following classification of recognition processes, according to their range of application.

\subsubsection{Fixed frequency setting: shape descriptors}

If only one frequency is accessible for the measurements, then the
only algorithm possible is classification based on shape
descriptors. Indeed, first-order PTs are not enough to
discriminate between objects. However, the use of shape
descriptors is limited to the twisted-ellipse case with nearly
full aperture (see Figs.~\ref{fig:results-SD-ellipse} and
~\ref{fig:results-aperture}, where some targets are not
recognized, such as the ellipse in
Fig.~\ref{fig:results-SD-ellipse} or the disk in
Fig.~\ref{fig:results-aperture}). Moreover, its stability with
respect to measurement noise is quite poor (see
Fig.~\ref{fig:stability-SD-mono}).

\subsubsection{Multifrequency setting: spectral content of PTs}

In the case of multifrequency measurements, shape descriptors do
not increase their stability enough compared to singular values
(see Fig.~\ref{fig:stability-SD-multi}). Hence, it is better to
use singular values of the PTs (see
Figs.~\ref{fig:results-multi-complete-banana} to
\ref{fig:results-multi-ratio-ellipse}). One can see that:
\begin{itemize}
\item considering all singular values
(Figs.~\ref{fig:results-multi-complete-banana} and
\ref{fig:results-multi-complete-ellipse}) is much more stable than
considering ratios of singular values
(Figs.~\ref{fig:results-multi-ratio-banana} and
\ref{fig:results-multi-ratio-ellipse}); \item the aperture does
not change very much the stability.
\end{itemize}
In this regard, the most stable algorithm is to recognize all
singular values when the fish is a twisted ellipse
(Fig.~\ref{fig:results-multi-complete-banana}), leading us to a
probability of detection superior to $90\%$ with noise level of
$125\%$. This is a huge gap when compared to the recognition
process with shape descriptors, allowing only a few percents of
noise. Note that the noise level is computed with respect to the
perturbation in the measurements $\mathbb{Q}$ given by
(\ref{eq:def-Q}), which is of order of the target volume, see
(\ref{eq:Qsr}). Hence, a noise level of $125\%$ remains small
compared to the actual transdermal potential $u$.

\subsubsection{Background field elimination}

We can see in Figs.~\ref{fig:results-imag-ellipse-complete} and
\ref{fig:results-imag-ellipse-ratio} that taking the imaginary
part of the measurements in order to avoid the computation of the
background field does not significantly decrease the stability of
the reconstruction based on spectral content. Since the
reconstruction of CGPTs is very fast, the background field
elimination technique would yield to real-time imaging.

\section{Concluding remarks}

In this paper, we have successfully exhibited the physical
mechanism underlying shape recognition and classification in
active electrolocation. We have shown that extracting generalized
polarization tensors from the data and comparing invariants with
those of learned elements in a dictionary yields a classification
procedure with a good performance in the full-view case and with
moderate measurement noise level. However, this shape descriptor
based classification is instable in the limited-view case and for
higher noise level. When measurements at multiple frequencies are
used, the stability of our classification approach is
significantly improved by using phase shifts and keeping only the
first-order polarization tensor. The resulting spectral induced
polarization based classification is very robust.

Our results open the door for the application of the extended
Kalman filter developed in~\cite{ammari2012tracking} to show the
feasibility of a tracking of both location and orientation of a
target from perturbations of the electric field on the skin
surface of the fish. It also remains to understand to what extent
the spectral induced polarization approach could help us retrieve
the electric parameters of the target or locate and recognize
multiple targets.

\begin{acknowledgments} This work was supported by the ERC Advanced Grant Project
MULTIMOD--267184. \end{acknowledgments}

\bibliographystyle{plain} \bibliographystyle{plain} \bibliographystyle{plain}
\bibliography{bib_tracking}

\appendix

In this appendix, we numerically illustrate the main findings in
this paper and show the potential of electro-sensing for shape
recognition and classification.

\begin{figure}
\centering
\includegraphics[width=8cm, height=4.6cm]{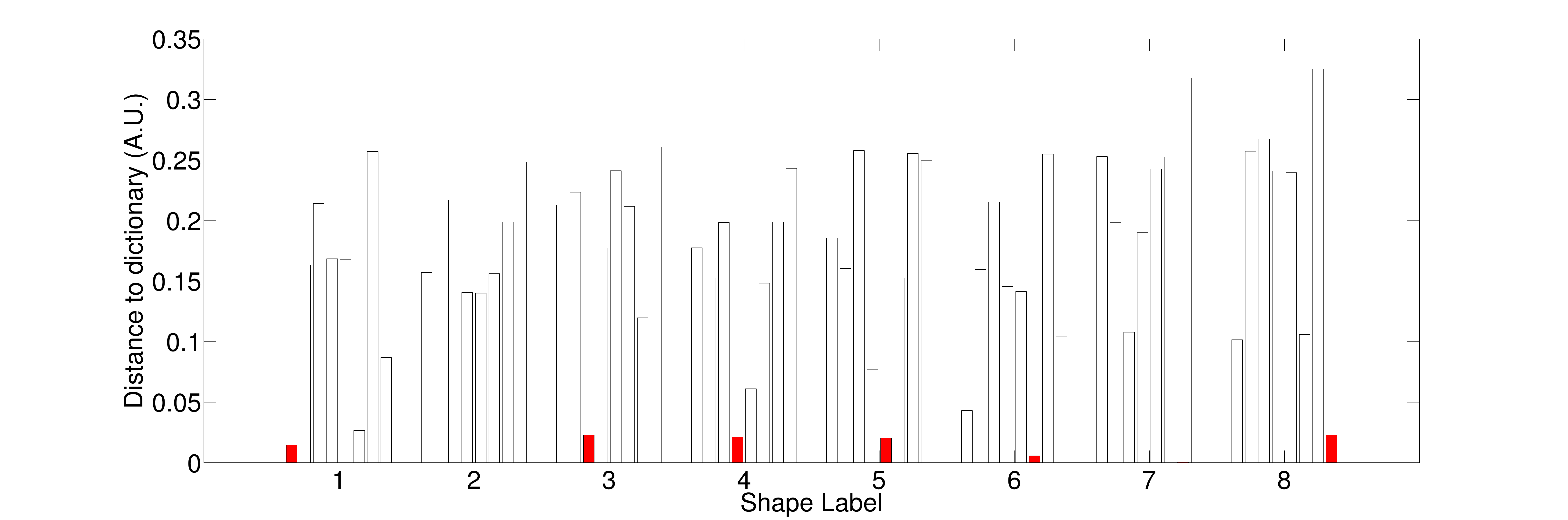}
%{banana_mono_0noise}
\caption{\label{fig:results-SD-banana}Results of the matching with
the dictionary (differences between shape descriptors), when the
fish is  a twisted ellipse. The red bar represents the bar of the
target being identified. In the $x$-coordinates, 1 stands for the
disk, 2 for the ellipse, 3 for the letter A, 4 for the letter E, 5
for the square, 6 for the rectangle, 7 for the triangle and 8 for
the ellipse with different electrical parameters. In the
$y$-coordinates, the distance between the shape descriptor of the
target -~computed from measurements~- and the shape descriptors of
the dictionary. }
\end{figure}

\begin{figure}
\centering
\includegraphics[width=8cm, height=4.6cm]{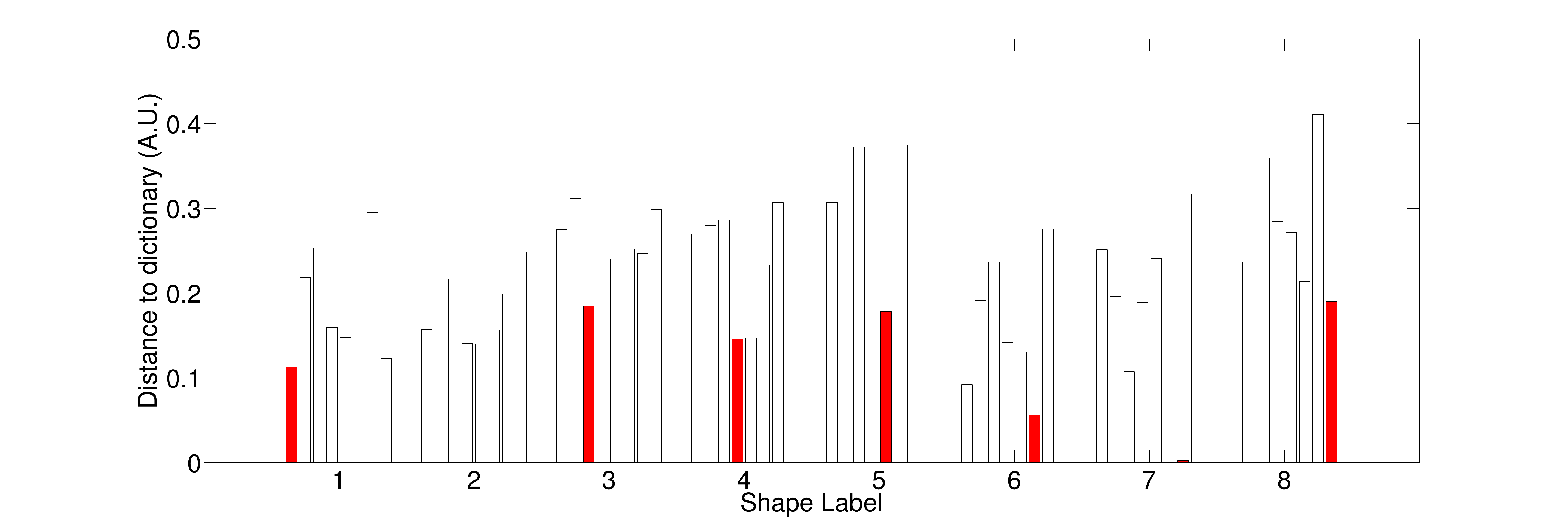}
\caption{Same results as in~\ref{fig:results-SD-banana}, with an
ellipse-shaped fish. \label{fig:results-SD-ellipse}}
\end{figure}

\begin{figure}
\centering
\includegraphics[width=8cm, height=4.6cm]{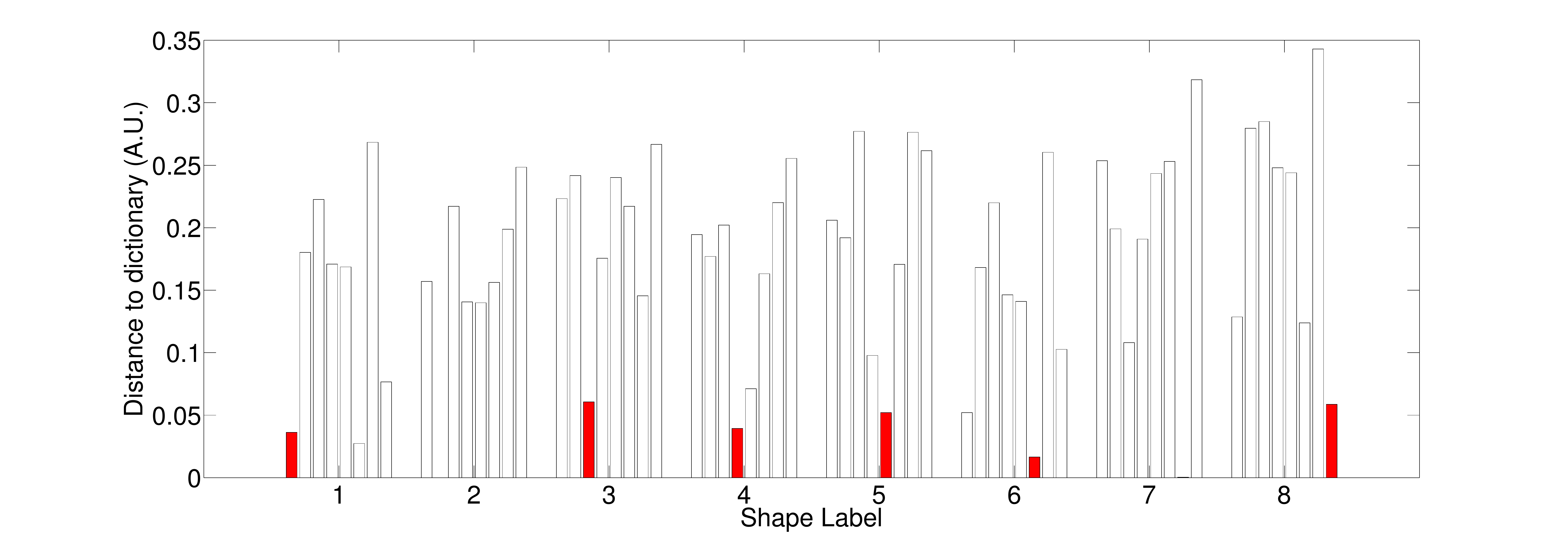}
\caption{\label{fig:results-aperture} Results of the matching with
the dictionary for a twisted ellipse shaped fish with smaller
aperture (see Fig.~\ref{fig:aperture}).}
\end{figure}

\begin{figure}
\centering
\includegraphics[width=8cm, height=4.6cm]{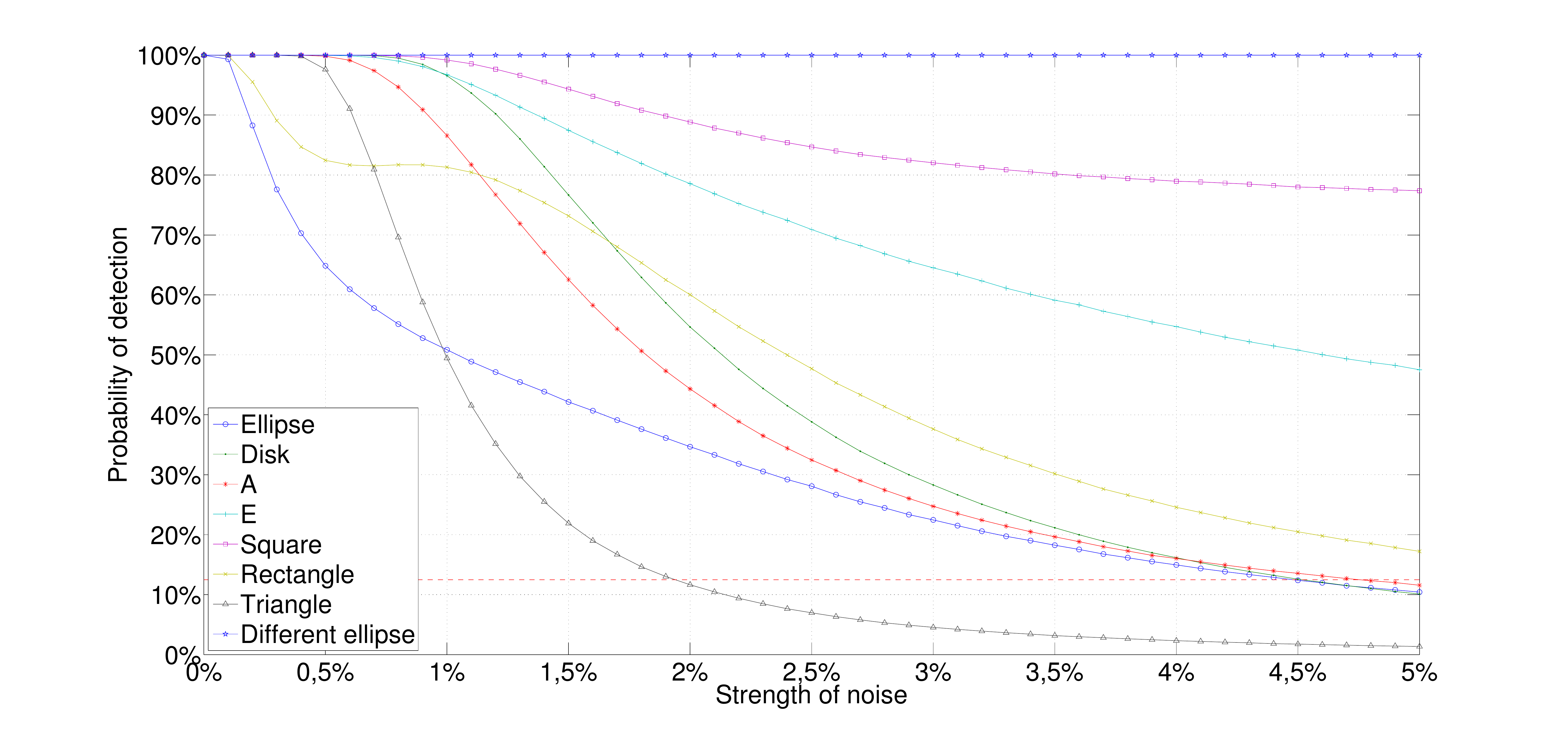}
\caption{\label{fig:stability-SD-mono}Stability of classification
based on Shape Descriptors, when the fish is a twisted ellipse.
Here, only one frequency (the smallest one) is considered. The
threshold of $12.5\%$ that corresponds to a randomly chosen target
is represented in red dotted line. Here,
$N_{\textrm{stabil}}=10^5$.}
\end{figure}

\begin{figure}
\centering
\includegraphics[width=8cm, height=4.6cm]{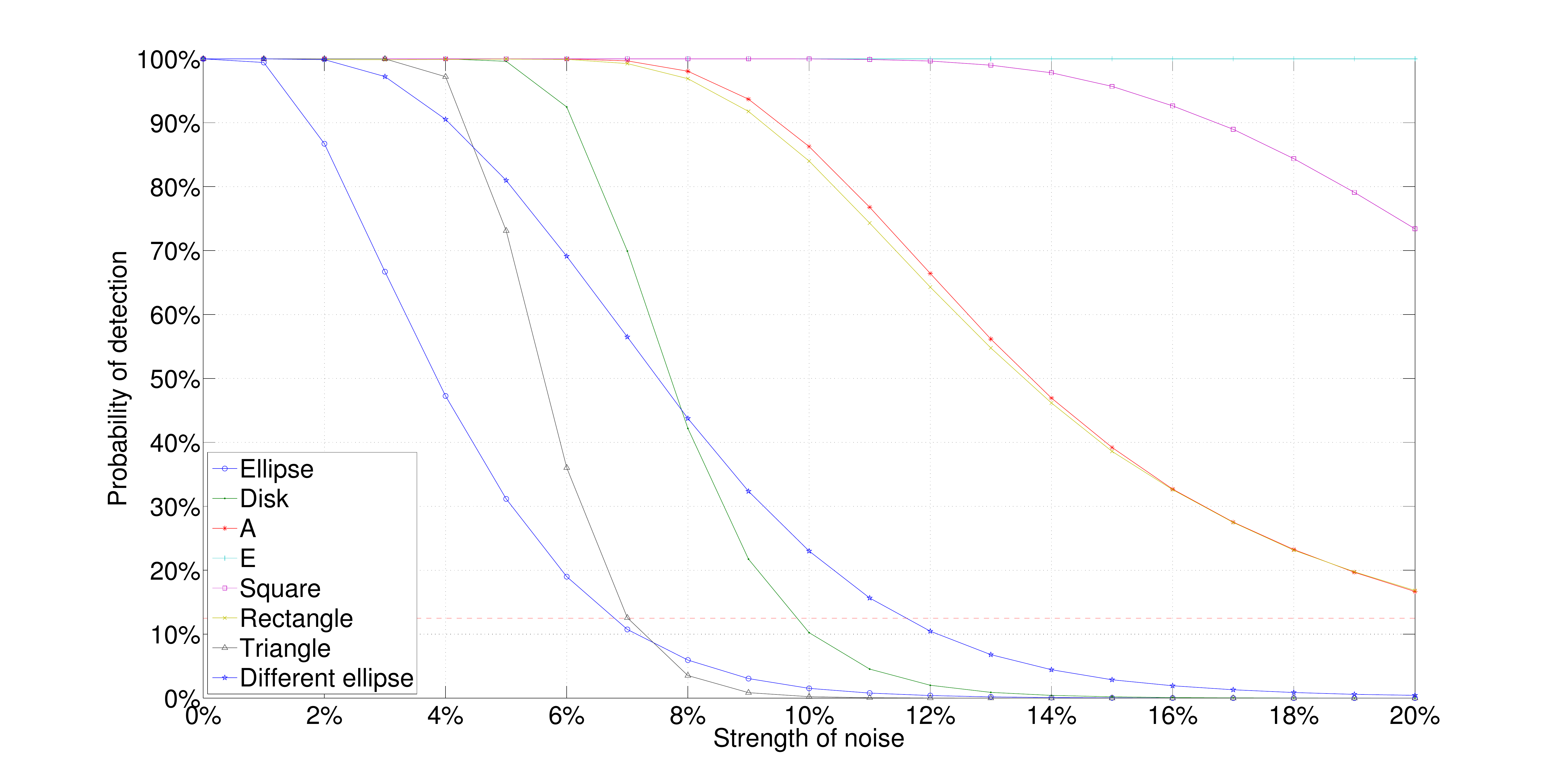}
\caption{\label{fig:stability-SD-multi}Stability of classification
based on multifrequency Shape Descriptors, when the fish is a
twisted ellipse. Here, $N_{\textrm{stabil}}=5.10^4$.}
\end{figure}

\begin{figure}
\centering
\includegraphics[width=8cm, height=4.6cm]{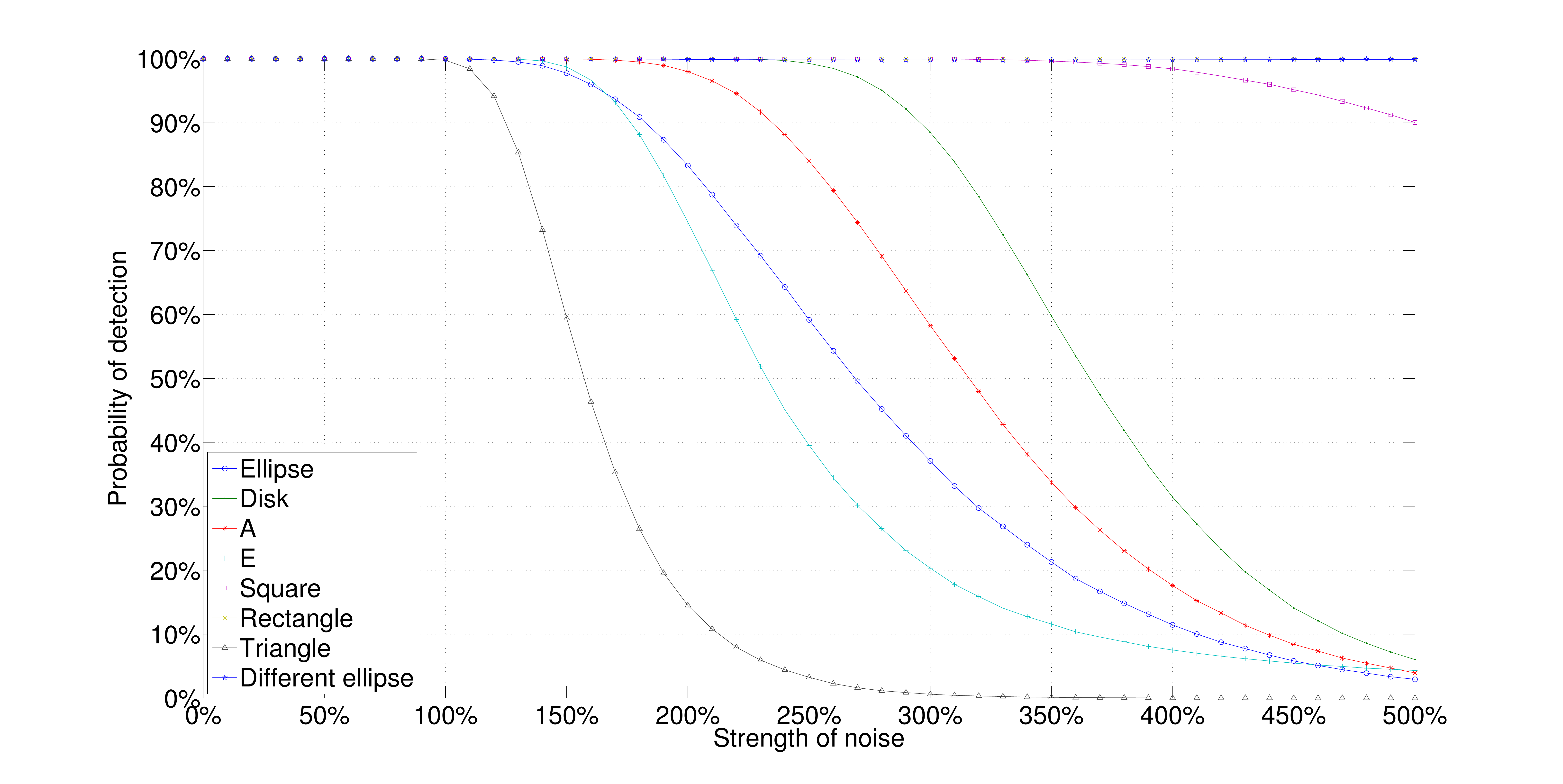}
\caption{\label{fig:results-multi-complete-banana}Stability of
classification based on differences between all singular values of
PTs, when the fish is a twisted ellipse. The characteristic size
of the target is supposed to be known. Here,
$N_{\textrm{stabil}}=5.10^4$.}
\end{figure}

\begin{figure}
\centering
\includegraphics[width=8cm, height=4.6cm]{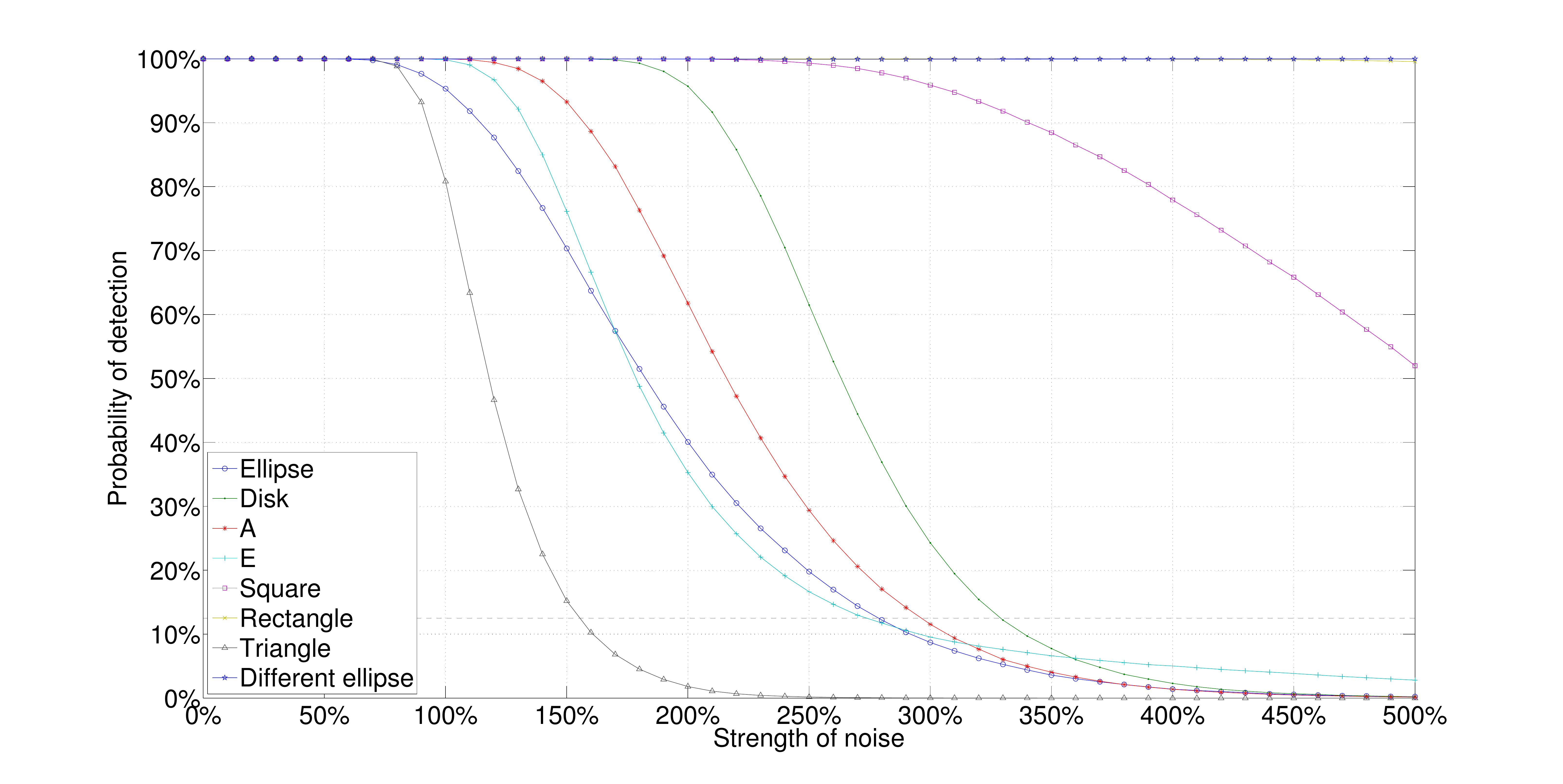}
\caption{\label{fig:results-multi-complete-ellipse}Stability of
classification based on differences between all singular values of
PTs, when the fish is a straight ellipse. The characteristic size
of the target is supposed to be known. Here,
$N_{\textrm{stabil}}=5.10^4$.}
\end{figure}

\begin{figure}
\centering
\includegraphics[width=8cm, height=4.6cm]{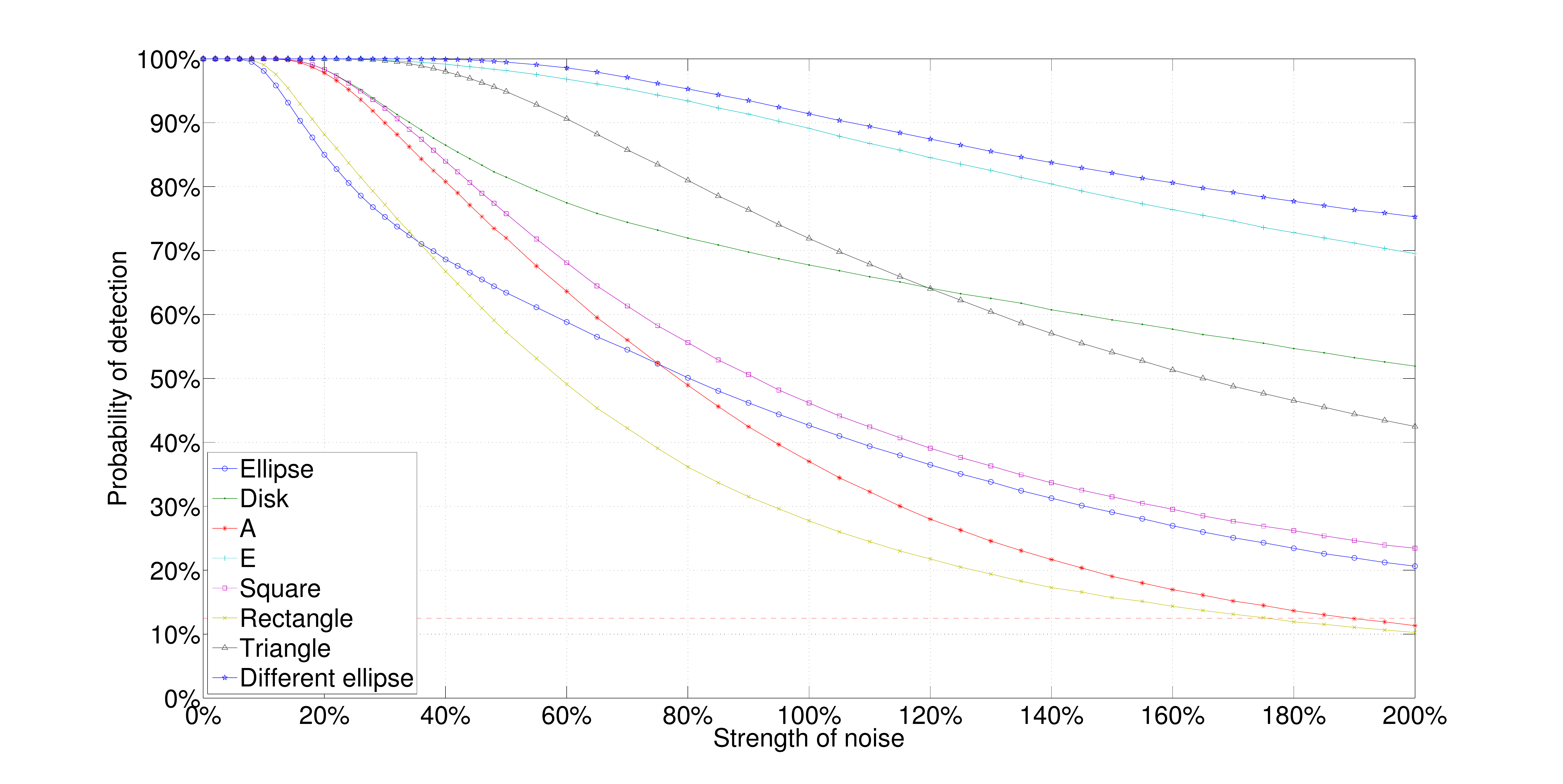}
\caption{\label{fig:results-multi-ratio-banana}Stability of
classification based on differences between ratios of singular
values, when the fish is a twisted ellipse. Here,
$N_{\textrm{stabil}}=5.10^4$.}
\end{figure}

\begin{figure}
\centering
\includegraphics[width=8cm, height=4.6cm]{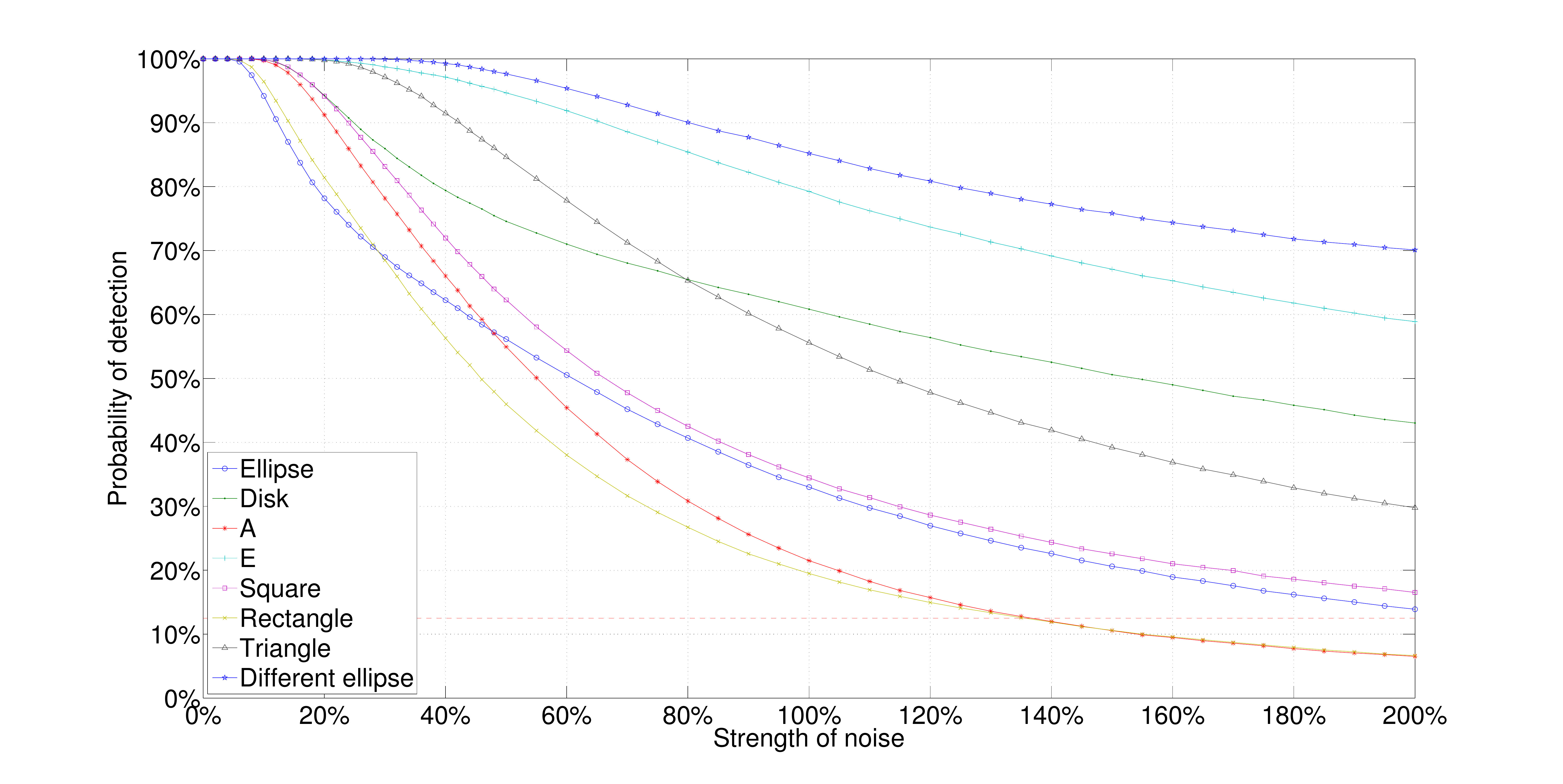}
\caption{\label{fig:results-multi-ratio-ellipse}Stability of
classification based on differences between ratios of singular
values, when the fish is a straight ellipse. Here,
$N_{\textrm{stabil}}=5.10^4$.}
\end{figure}

\begin{figure}
\centering
\includegraphics[width=8cm, height=4.6cm]{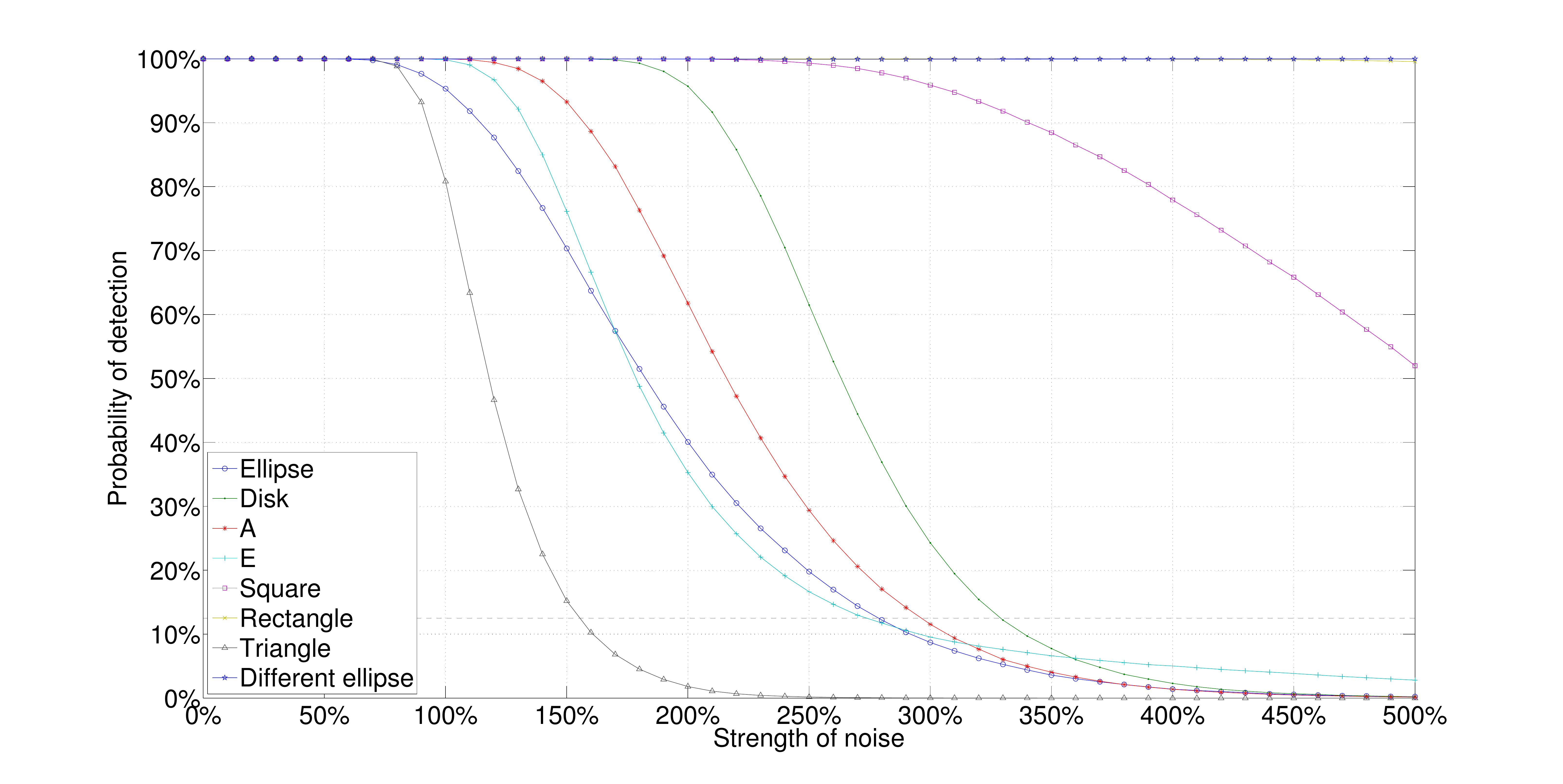}
\caption{\label{fig:results-imag-ellipse-complete}Classification
with imaginary part of the PT, when the fish is a straight
ellipse. All singular values are considered to discriminate
between the targets. Here, $N_{\textrm{stabil}}=10^5$.}
\end{figure}

\begin{figure}
\centering
\includegraphics[width=8cm, height=4.6cm]{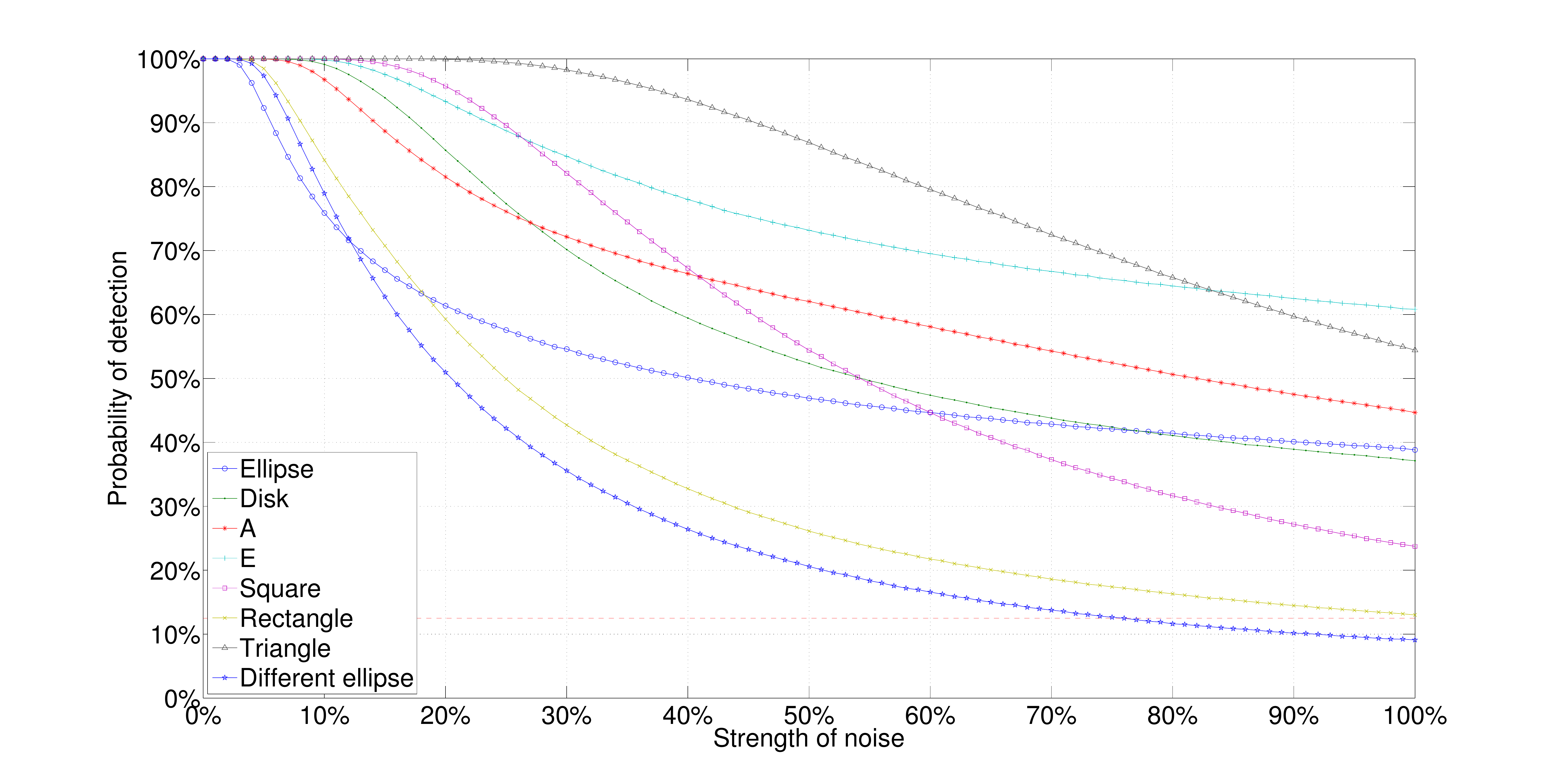}
\caption{\label{fig:results-imag-ellipse-ratio}Classification with
imaginary part of the PT, when the fish is a straight ellipse.
Ratios of singular values are considered to discriminate between
the targets. Here, $N_{\textrm{stabil}}=10^5$.}
\end{figure}

\end{article}
\end{document}